\documentclass[journal,twoside,web]{ieeecolor}
\usepackage{generic}
\usepackage{cite}
\usepackage{amsmath,amssymb,amsfonts}
\usepackage{algorithmic}
\usepackage{graphicx}
\usepackage{algorithm,algorithmic}
\usepackage{hyperref}
\hypersetup{hidelinks=true}
\usepackage{textcomp}
\def\BibTeX{{\rm B\kern-.05em{\sc i\kern-.025em b}\kern-.08em
    T\kern-.1667em\lower.7ex\hbox{E}\kern-.125emX}}
\usepackage{mathtools}
\usepackage{epstopdf}

\newtheorem{theorem}{Theorem}
\newtheorem{assumption}{Assumption}
\newtheorem{corollary}[theorem]{Corollary}
\newtheorem{lemma}[theorem]{Lemma}

\newtheorem{definition}{Definition}
\newtheorem{remark}{Remark}

\newcommand{\ie}{\textit{i.e.}}
\newcommand{\eg}{\textit{e.g.}}
\DeclareMathOperator{\VEC}{vec}

\DeclareMathOperator{\rank}{rank}
\DeclareMathOperator{\diag}{diag}
\DeclareUnicodeCharacter{221E}{$\inf$}
\newcommand{\RR}{\mathbb{R}}

\newcommand{\XX}{\mathcal{X}}
\newcommand{\QQQ}{\mathcal{Q}}
\newcommand{\SSS}{\mathcal{S}}
\newcommand{\RRR}{\mathcal{R}}
\newcommand{\PP}{\mathcal{P}}

\newcommand{\mc}{\mathcal}

\begin{document}
\title{Data-Driven Model Reduction by Moment Matching for Linear and Nonlinear Parametric Systems}
\author{Hanqing Zhang, \IEEEmembership{Student Member, IEEE}, Junyu Mao, \IEEEmembership{Student Member, IEEE}, Mohammad Fahim Shakib, \IEEEmembership{Member, IEEE}, and Giordano Scarciotti, \IEEEmembership{Senior Member, IEEE}
\thanks{The authors are with the Department of Electrical and Electronic Engineering, Imperial College London, SW7 2AZ London, U.K. (e-mail: hanqing.zhang21@imperial.ac.uk; junyu.mao18@imperial.ac.uk; m.shakib@imperial.ac.uk; g.scarciotti@imperial.ac.uk). }}

\maketitle

\begin{abstract}
Theory and methods to obtain parametric reduced-order models by moment matching are presented. The definition of the parametric moment is introduced, and methods (model-based and data-driven) for the approximation of the parametric moment of linear and nonlinear parametric systems are proposed. These approximations are exploited to construct families of parametric reduced-order models that match the approximate parametric moment of the system to be reduced and preserve key system properties such as asymptotic stability and dissipativity. The use of the model reduction methods is illustrated by means of a parametric benchmark model for the linear case and a large-scale wind farm model for the nonlinear case. In the illustration, a comparison of the proposed approximation methods is drawn and their advantages/disadvantages are discussed.
\end{abstract}


\section{Introduction}
\label{sec:introduction}
\IEEEPARstart{H}{igh-fidelity} mathematical dynamical models are of great importance for control analysis and design~\cite{benner2015survey}. Such high-fidelity models often tend to be of high order and be described by numerous ordinary differential equations, and thus are computationally expensive and time-consuming to simulate and analyse~\cite{antoulas2005approximation,scarciotti2017data}. Examples include models for power systems~\cite{scarciotti2016low}, weather forecast models~\cite{antoulas2005approximation}, and models for differential systems that arise from finite element analysis~\cite{scarciotti2024interconnection}. A solution to this intensive computational problem is model reduction, which aims to obtain a ``simpler", \eg, with a lower number of states, reduced-order model that approximates the behaviour of the original full-order model.

Several model reduction techniques have been developed for both linear and nonlinear systems. For linear systems, two main groups of techniques play a significant role: singular value decomposition (SVD) based approximation methods and Krylov-based approximation methods (also known as moment matching methods) \cite{antoulas2005approximation}. SVD-based methods include balanced truncation \cite{moore1981principal}, Hankel-norm approximations \cite{glover1984all}, and proper orthogonal decomposition (POD) \cite{willcox2002balanced}. Moment matching methods, which yield a reduced-order model whose transfer function (and derivatives of this) takes the same values (\ie, moments) as the transfer function (and derivatives of this) of a given large-scale system at specific frequencies (also called interpolation points), include the rational interpolation method \cite{gallivan2004sylvester,gugercin2008h_2}, the interconnection-based framework \cite{astolfi2010model,scarciotti2017nonlinear,scarciotti2024interconnection}, and the Loewner framework \cite{mayo2007framework}. The aforementioned methods have been enhanced and extended to nonlinear systems, see, for instance, \cite{scherpen2002nonlinear} for nonlinear balancing, \cite{willcox2002balanced} for enhancements of POD, and \cite{simard2021nonlinear} for the nonlinear Loewner framework.

However, most traditional model reduction techniques lose efficacy for systems whose governing differential equations are dependent on a set of parameters. The reason lies in that most of the methods are often intrusive, meaning that they require the information of the original full-order models. Thus, it is unaffordable and undesirable in terms of computing power to derive a reduced-order model for every possible value of each parameter. Parametric model reduction resolves this problem by obtaining a parametric reduced-order model which approximates the behaviour of the original full-order model over a range of potential values of the parameters. Such reduction methods have been extensively investigated since the '00s, because of their critical role in design, control, optimisation, and uncertainty quantification settings~\cite{benner2015survey}. The development of parametric model reduction methods has additionally been motivated by several applications, \eg, large-scale systems parametrised by material elasticity~\cite{lieu2007adaptation}, physical geometry~\cite{daniel2004multiparameter}, and boundary conditions~\cite{feng2005preserving}.

For linear parametric systems, parametric model reduction techniques are obtained by extending moment matching methods to the parametric case, see~\cite{feng2005preserving,gunupudi2002analysis,daniel2004multiparameter}. These methods consider moment as the value of the transfer function with respect to not only the Laplace variable but also the model parameter. Another method is matrix interpolation (also known as manifold interpolation) \cite{panzer2010parametric}. However, it has been shown in \cite{schwerdtner2022structured} that the obtained parametric reduced-order model only approximates the dynamics of the full-order model accurately at the parameter interpolation points, and often results in large approximation errors at values between the interpolation points. Techniques that extend balanced truncation \cite{wittmuess2016parametric}, the Loewner framework \cite{lefteriu2011parametric}, and the optimal interpolation method \cite{mlinaric2023optimal} to linear parametric systems are mentioned here for the sake of completeness. 

For the reduction of nonlinear parametric systems, a recent overview can be found~\cite{goyal2024dominant}, which extends linear interpolatory methods to certain classes of polynomial structured parametric systems. Such an extension is achieved by introducing the generalised multivariable transfer functions induced by the Volterra series representation of these systems. 
To deal with more general classes of nonlinear parametric systems, trajectory piecewise linear methods have been proposed in~\cite{bond2007piecewise}. 
Note, however, that this method essentially reduces the problem to a linear parametric model reduction problem and does not deal with the nonlinearity directly, potentially limiting its global effectiveness. 
In addition, POD-based methods (often combined with the discrete empirical interpolation method (DEIM)), which address system nonlinearity directly, have been shown to be effective in certain nonlinear parametric cases, see, \eg,~\cite{afkham2017structure}. 
Note that these methods typically require access to not only input-output data but also state trajectory data.

\textbf{Contributions.} The main contributions of this article are summarised here. (1)~We define the parametric moment of a linear or nonlinear parametric system based on which a family of parametric reduced-order models is constructed. (2)~Since the parametric moment is often difficult and possibly intractable to obtain, we propose two methods, via series expansion and basis functions, to approximate the parametric moment. (3)~In addition, we propose a data-driven enhancement of the method for settings in which the knowledge of the full-order model or the parameter dependency is unavailable. (4)~Using the approximate parametric moment, we construct a family of parametric reduced-order models that match the approximate parametric moment of the original parametric system and, simultaneously, can preserve specific properties, notably asymptotic stability and dissipativity. 

The advantages of our parametric model reduction framework are manifold. (1)~The framework achieves parametric model reduction for both linear and nonlinear parametric systems. (2)~Extra freedom is provided by which the resulting parametric reduced-order model can retain specific properties of the original parametric system, \eg, asymptotic stability and dissipativity. (3)~The method via series expansion provides high-fidelity parametric reduced-order models that match the moments with a guaranteed accuracy within a certain neighbourhood of the expansion point on the parameter space. (4)~In contrast, the method via basis functions generates parametric reduced-order models which exhibit a uniformly small error across the entire parameter space. (5)~The (nonlinear) data-driven version of the method via basis functions has computational complexity proportional to ``the reduced order'' times ``the number of parameter sampling points''. This is computationally efficient compared to intrusive model reduction methods while, with respect to other data-driven methods, \eg, POD-based methods, DEIM, and neural networks, only (noisy) input-output data are required. 

A preliminary small part of this work is under review for presentation at the 2025 conference on decision and control (CDC 2025)~\cite{Zhang2025Data}, addressing only the stability-preserving parametric model reduction problem for linear parametric systems using a series expansion approximation of the moment.
For linear parametric systems, the current paper extends~\cite{Zhang2025Data} to the approximation of the moment using basis functions and to the dissipativity-preserving problem.
The parametric moment matching method for nonlinear parametric systems and the data-driven enhancements, both in the linear and nonlinear parametric settings, have not been presented elsewhere.

\textbf{Organisation.} The rest of the paper is organised as follows. 
Section~\ref{sec:momwntnMomentMatching} recalls the notion of moment and moment matching for linear non-parametric systems as developed in~\cite{astolfi2010model,scarciotti2017nonlinear}, while Section~\ref{sec:paraMomentMatching} extends the results to the linear parametric case by introducing the notion of parametric moment and constructing a family of parametric reduced-order models. 
Subsequently, a subclass of this family that is stability or dissipativity preserving is presented in Section~\ref{sec:dissipativity}.
In Section~\ref{sec:MMforLinearParametricSys}, two methods are developed to approximate the parametric moment of a given linear parametric system. Then, the resulting families of parametric reduced-order models are constructed. In particular, Section~\ref{sec:MatrixBased} presents the method via series expansion while Section~\ref{sec:dataDriven} proposes the method via basis functions (with its data-driven version in Section~\ref{sec:linearDataDriven}). The effectiveness of these two methods is illustrated in Section~\ref{sec:linearExample} using a parametric benchmark model. The framework is enhanced in Section~\ref{sec:MMforNonlinearParametricSys} to achieve parametric model reduction for nonlinear parametric systems. In addition, the nonlinear enhancement is demonstrated using a large-scale model of a wind farm. Finally, Section~\ref{sec:conclusion} contains some concluding remarks.

\textbf{Notation.} We use standard notation. The symbols $\mathbb{R}$ and $\mathbb{C}$ denote the set of real numbers and complex numbers, respectively; $\RR_{>0}$ denotes the set of positive real numbers; and $\mathbb{C}_{<0}$ ($\mathbb{C}_0$) denotes the set of complex numbers with negative (zero) real part. The symbol $I_n$ denotes the identity matrix of dimension $n \times n$, and $\sigma(A)$ denotes the spectrum of the matrix $A \in \mathbb{R}^{n \times n}$. The symbol $\otimes$ indicates the Kronecker product and the superscript $\top$ denotes the transposition operator. For a vector $x \in \mathbb{R}^{n}$, $\|x\|_2$ indicates its Euclidean norm; and the vectorisation of a matrix $A \in \mathbb{R}^{n \times m}$ is denoted by $\VEC(A)$, which is a $nm \times 1$ vector obtained by stacking the columns of the matrix $A$ one on top of the other. The symbol $\iota$ denotes the imaginary unit.
For a symmetric matrix $A \in \mathbb{R}^{n \times n}$, $A \succ 0$ ($A \prec 0$) denotes that $A$ is positive- (negative-) definite.
In symmetric block matrices, the element $\star$ is induced by transposition. The symbol $\mathcal{O}(\cdot)$ denotes the big O notation.

\section{Moment Matching for Linear Parametric Systems\label{sec:preliminaries}}
In this section, we first recall the notion of moment and moment matching for linear non-parametric systems, as proposed in~\cite{astolfi2010model,scarciotti2017nonlinear}, in Section~\ref{sec:momwntnMomentMatching}. Then we extend these results to linear parametric systems in Section~\ref{sec:paraMomentMatching} by introducing the notion of parametric moment.
Finally, in Section~\ref{sec:dissipativity}, we present a method that preserves the stability or dissipativity properties of the original parametric system for the parametric reduced-order model.

\subsection{Notion of Moment and Moment Matching}\label{sec:momwntnMomentMatching}
Consider a linear time-invariant (LTI), single-input single-output (SISO), continuous-time parametric system described by
\begin{equation}
\begin{aligned}
    \Dot{x}{(t,p)} &= A(p)x{(t,p)} + B(p)u(t),\\
    y{(t,p)} &= C(p)x{(t,p)},\label{eq:sys}
\end{aligned}
\end{equation}
with parameter $p \in \mathcal{P} \subseteq \mathbb{R}$, state $x{(t,p)} \in \mathbb{R}^{n}$, input $u(t) \in \mathbb{R}$, output $y{(t,p)} \in \mathbb{R}$, and $A : \mathcal{P} \to \mathbb{R}^{n \times n}$, $B : \mathcal{P} \to \mathbb{R}^{n}$ and $C : \mathcal{P} \to \mathbb{R}^{1 \times n}$ {matrix-valued mappings}. Let $W{(s,p)} \coloneqq C(p)(sI - A(p))^{-1}B(p)$ be the associated transfer function.

In this section, we consider system \eqref{eq:sys} with a fixed value $p = p^*$. Thus, the resulting notion of moment is exactly identical\footnote{Of course, in this case the dependency on $p$ could be omitted. However, we keep this as we use the section to also introduce the parametric notation.} to that given in \cite[Chapter~11]{antoulas2005approximation}.

\begin{definition} \label{def:moments}
    Let $s_i \in \mathbb{C} \backslash \sigma(A(p^*))$. The $0$-moment of system \eqref{eq:sys} at $s = s_i$ is $\eta_0(s_i, p^*) = C(p^*)(s_i I - A(p^*))^{-1}B(p^*)${. The} $k$-moment of system \eqref{eq:sys} at $s = s_i$ is
    \begin{equation}
        \eta_k(s_i, p^*) = \frac{(-1)^k}{k!} \biggl[\frac{d^k}{ds^k} W(s, p^*)\biggr]_{s = s_i}, \label{eq:moment}
    \end{equation}
    with integer $k \geq 1$.  
\end{definition}
 
The moments of system \eqref{eq:sys} at $p = p^*$ can be characterised in terms of the solution of a Sylvester equation \cite{gallivan2004sylvester, gallivan2006model}. Based on this observation, it has been noted that the moments have a one-to-one relationship with the steady-state output response (provided it exists) of system~\eqref{eq:sys} at $p = p^*$ when excited by a certain signal generator~\cite{scarciotti2017nonlinear}, as shown in the following result.

\begin{theorem}[\!\!\cite{astolfi2010model}]\label{thm:moment2steadystate}
    Consider system \eqref{eq:sys}. Suppose $s_i \in \mathbb{C} \backslash \sigma(A(p^*))$ for all $i = 0, \dots, \eta$. Let $S \in \mathbb{R}^{\nu \times \nu}$ be a non-derogatory matrix with the characteristic polynomial ${p(s)} = \prod_{i = 0}^{\eta} (s-s_i)^{k_i}$, where $\nu = \sum_{i = 0}^{\eta} k_i$, and $L \in \mathbb{R}^{1 \times \nu}$ be such that $(S, L)$ is observable. Then, there exists a one-to-one relationship between the moments of {system} \eqref{eq:sys} at the eigenvalues of $S$, namely
    \begin{multline*}
        \eta_0(s_0, p^*), \dots, \eta_{k_0 - 1}(s_0, p^*), \dots, \\ \eta_0(s_\eta, p^*), \dots, \eta_{k_{\eta} - 1}(s_\eta, p^*)
    \end{multline*}
    and
    \begin{enumerate}
        \item the matrix $C(p^*)\Pi(p^*)$, where $\Pi(p^*) \in \mathbb{R}^{n \times \nu}$ is the unique solution of the Sylvester equation
        \begin{equation} 
            A(p^*)\Pi(p^*) + B(p^*)L = \Pi(p^*)S;\label{eq:SylvesterEq}
        \end{equation}
        \item \label{thm:moment2steadystateP2} the steady-state response, provided $\sigma(A(p^*)) \subset \mathbb{C}_{<0}$ and $\sigma(S) \subset \mathbb{C}_0$ simple, of the output $y(t,p^*)$ of system~\eqref{eq:sys} at $p=p^*$ interconnected with the signal generator 
        \begin{equation}\label{eq:SignalGenerator}
            \Dot{\omega}(t) = S\omega(t), \quad u(t) = L\omega(t),
        \end{equation}
        for any $\omega(0)$ such that $(S,\omega(0))$ is excitable\footnote{See \cite[Definition~2]{padoan2017geometric} for the definition of an excitable pair.}. Moreover, the interconnected system has a global invariant manifold described by {$\mathcal{M}(p^*) = \{(x,\omega) : x=\Pi(p^*)\omega\}$}. Hence, for all $t \in \mathbb{R}$,
        \begin{equation}\label{eq:manifold}
            x(t,p^*) = \Pi(p^*)\omega(t) + \epsilon(t,p^*), 
        \end{equation}
        with $\epsilon(t,p^*)=e^{A(p^*)t}(x(0,p^*)-\Pi(p^*)\omega(0))$ an exponentially decaying signal.\hfill $\blacksquare$
    \end{enumerate} 
\end{theorem}

A family of models that match the moments at $(S,L)$ (that is, models that have the same moments as \eqref{eq:sys} at $\sigma(S)$, as described in Theorem~\ref{thm:moment2steadystate}) is given by
\begin{equation}\label{eq:rom}
    \begin{aligned}
        \Dot{\xi}(t,p^*) &= (S - {G(p^*)}L)\xi(t,p^*) + {G(p^*)}u(t),\\
        \psi(t,p^*) &= C(p^*)\Pi(p^*)\xi(t,p^*),
    \end{aligned}
\end{equation}
and parametrised by any ${G(p^*)}$ such that $\sigma(S - {G(p^*)}L) \cap \sigma(S) = \emptyset$. System~\eqref{eq:rom} is a reduced-order model of system~\eqref{eq:sys} at $(S, L)$ and $p = p^*$ if $\nu < n$.
We refer to~\cite{shakib2023time} for moment matching for multiple-input, multiple-output systems.

\subsection{Notion of Parametric Moment and Moment Matching}\label{sec:paraMomentMatching}
If \eqref{eq:SylvesterEq} has a unique solution for each $p \in \mathcal{P}$, the moments of system \eqref{eq:sys} at $(S, L)$ and $p$ can be characterised in terms of the matrix $C(p)\Pi(p)$ as a function of $p$. Thus, the notion of moment of system \eqref{eq:sys} at a fixed $p$ (as introduced in Section~\ref{sec:momwntnMomentMatching}) can be extended to that of parametric moment on a set $\mathcal{P}$. We formalise first the required assumptions and then introduce the notion of parametric moment.

\begin{assumption}\label{ass:obserAndNoCommonEig}
    The signal generator \eqref{eq:SignalGenerator} is observable and does not share common eigenvalues with system \eqref{eq:sys} for all $p \in \mathcal{P}$.
\end{assumption}

Note that $\Pi(p^*)$ is the unique solution of \eqref{eq:SylvesterEq} if $\sigma(S) \cap \sigma(A(p^*)) = \emptyset$. Thus, under Assumption \ref{ass:obserAndNoCommonEig}, $\sigma(S) \cap \sigma(A(p)) = \emptyset$ for all $p \in \mathcal{P}$ ensures that $\Pi(\cdot)$ exists and is unique on $\mathcal{P}$. Consequently, the parametric moment is well defined.

\begin{definition}\label{def:ParametricMoments}
    We call $\overline{C\Pi} \coloneqq C(\cdot)\Pi(\cdot): \mathcal{P} \to \mathbb{R}^{1 \times \nu}$ the \textit{parametric moment} of system \eqref{eq:sys} at $(S, L)$ on $\mathcal{P}$, where $\Pi: \mathcal{P} \to \mathbb{R}^{n \times \nu}$ is a unique mapping solving \eqref{eq:SylvesterEq} for all $p \in \mathcal{P}$.
\end{definition}

With Definition \ref{def:ParametricMoments}, an associated family of parametric reduced-order models matching the parametric moment of system \eqref{eq:sys} at $(S,L)$ {on} $\mathcal{P}$ is given by
\begin{equation}\label{eq:prom}
    \begin{aligned}
        \Dot{\xi}(t,p) &= (S - {G(p)}L)\xi(t,p) + {G(p)}u(t),\\
        \psi(t,p) &= {\overline{C\Pi}}(p)\xi(t,p),
    \end{aligned}
\end{equation}
with $G: \mathcal{P} \to \mathbb{R}^\nu$ any mapping such that\footnote{This spectrum condition ensures that the Sylvester equation associated to~\eqref{eq:prom} has a unique solution for all $p \in \mathcal{P}$, for the same reasons described above.} $\sigma(S - {G(p)}L) \cap \sigma(S) = \emptyset$ for all $p \in \mathcal{P}$ and $\nu < n$.
\begin{remark}
    The mapping $G(\cdot)$ has $\nu$ free entries which allow enforcing specific properties onto the parametric reduced-order model, such as setting prescribed eigenvalues or zeros, at desired values of $p$, see \cite{astolfi2010model}.
\end{remark}

\begin{remark}
    The matrix $S$ could be constructed in general as a matrix-valued function of $p$, whose eigenvalues vary with different values of $p$, to reduce the approximation error across the parameter space $\mathcal{P}$.
    In this article, we consider $S$ and $L$ to be constant for ease of notation and demonstration. For instance, this simplification has the advantage that only a constant $G(p) = \bar G$ is needed to enforce, for example, the stability of the parametric reduced-order model \eqref{eq:prom}, for all $p \in \mathcal{P}$.
\end{remark}

\subsection{Stability- and Dissipativity-Preserving Reduction}\label{sec:dissipativity}

Models constructed through moment matching preserve the behaviour of the original system near the interpolation points.
In addition, it might also be desired to preserve stability and dissipativity properties of the original system.
While these are fundamental properties of a dynamical system, preserving them is challenging as this needs to be done for all parameter values of interest, namely on $\PP$.
While stability-preserving model reduction techniques for linear parametric systems were presented in, \eg{},~\cite{eid2011stability,baur2011parameter,beattie2019sampling}, not necessarily in the scope of moment matching, the preservation of dissipativity remains an open problem in the literature (whereas for systems without parametric dependency, a stability- and/or dissipativity-preserving model reduction method was presented in~\cite{shakib2025dissipativity}).

The family of models in~\eqref{eq:prom} is parametrised by the mapping $G: \PP \rightarrow \RR^{\nu}$. 
In this section, we address the problem of selecting this mapping such that stability and/or dissipativity of the original system is preserved for the reduced-order model across $\PP$.
We start by defining the notion of $(\QQQ,\SSS,\RRR)$-dissipativity for the parameter-dependent system~\eqref{eq:sys}.
This definition is an extension of the standard $(\QQQ,\SSS,\RRR)$-dissipativity notion for systems without parametric dependency, see~\cite[Definition~4]{kottenstette2014relationships}.

\begin{definition}\label{def:QSR}    
    System~\eqref{eq:sys} is \emph{dissipative} with respect to the energy supply rate $s : \RR \times \RR \times \PP \to \RR : (u,y,p) \mapsto s(u,y,p)$ and for any $p\in\PP$, if there exists a non-negative storage function $V:\RR^{n}\times \PP \rightarrow \RR$, such that for all inputs signal $u:\RR\to\RR$, all trajectories $x: \RR\to\RR^n$, and all $t_2\ge t_1$, the following inequality holds
    \begin{equation}\label{eq:diss_ineq}
        V(x(t_2,p),p) \le V(x(t_1,p),p)  + \int_{t_1}^{t_2} s(u(t),y(t,p),p)\textrm{d}t
    \end{equation}
    for all $p\in\PP$.
    Additionally, the system is $(\QQQ,\SSS,\RRR)$\emph{-dissipative} for all $p\in\PP$ if it is dissipative with respect to the quadratic supply rate
    \begin{equation}\label{eq:quadratic_supply_rate}
        s(u,y,p) = y^\top\QQQ(p) y + 2 u^\top\SSS(p) y + u^\top\RRR(p) u,
    \end{equation}
    for all $p\in \PP$,
    where $\QQQ : \PP\rightarrow \RR$, $\SSS : \PP\rightarrow \RR$, and $\RRR: \PP\rightarrow \RR$.
\end{definition}

Special cases of $(\QQQ,\SSS,\RRR)$-dissipativity include the notions of passivity and $L_2$-gain:
\begin{itemize}
    \item If~\eqref{eq:diss_ineq} with $s$ in~\eqref{eq:quadratic_supply_rate} holds with $\QQQ(p)=0$, $\SSS(p) = \frac{1}{2}$, and $\RRR(p)=0$, then system~\eqref{eq:sys} is \emph{passive} for all $p\in\PP$.
    \item If~\eqref{eq:diss_ineq} holds with $\QQQ(p)=-1$, $\SSS(p) = 0$, and $\RRR(p)=\gamma^2(p)$, then the system is $L_2$-stable from $u$ to $y$ with $L_2$-gain $\gamma: \PP \to \RR_{>0}$ for all $p\in\PP$.
\end{itemize}

We now present conditions for asymptotic stability and $(\QQQ,\SSS,\RRR)$-dissipativity of the parametric system~\eqref{eq:sys}.

\begin{lemma}\label{lem:stab_diss}
The following statements are true.
\begin{itemize}
    \item System~\eqref{eq:sys} is asymptotically stable for all $p\in\PP$ if and only if there exists a matrix-valued function~$\XX: \PP \rightarrow \RR^{n\times n}$ such that
    \begin{equation}\label{eq:stab_LMI}
        A^\top(p)\XX(p) + \XX(p) A(p) \prec 0, \, \XX(p) \succ 0, \, \forall p\in \PP.
    \end{equation}
    \item Given matrix functions $\QQQ: \PP\rightarrow \RR$, $\SSS: \PP\rightarrow \RR$, and $\RRR:\PP\rightarrow \RR$, system~\eqref{eq:sys} is $(\QQQ,\SSS,\RRR)$-dissipative for all $p\in\PP$ if and only if there exists a matrix-valued function~$\XX: \PP \rightarrow \RR^{n\times n}$, such that
    \begin{multline}\label{eq:diss_LMI}
        \begin{bmatrix}
            A^\top(p) \XX(p) + \XX(p) A(p) & \star  \\
            B^\top(p) \XX(p) & 0 \\
        \end{bmatrix}\\
        - \begin{bmatrix}
            C^\top(p) \QQQ(p) C(p) & \star \\ \SSS(p)C(p) & \RRR(p)
        \end{bmatrix}
        \preceq 0, \\
        \XX(p) = \XX^\top(p) \succ 0, \, \forall p\in \PP.
    \end{multline}
\end{itemize}
\end{lemma}
\begin{proof}
The proof of the first statement follows from~\cite[Theorem $5.19$]{antoulas2005approximation}, which states for a fixed $p\in\mathcal{P}$ that $A(p)$ is a Hurwitz matrix if and only if the inequality~\eqref{eq:stab_LMI} holds.
The proof of the second statement follows from \cite[Lemma~2]{kottenstette2014relationships}, which for a fixed $p \in \mathcal{P}$ states that the system is $(\QQQ,\SSS,\RRR)$-dissipative if and only if the inequality~\eqref{eq:diss_LMI} holds.
\end{proof}
\medskip

Given that a matrix-valued function $\XX(\cdot)$ is found, either for stability or for dissipativity, we now present a selection of the mapping $G(\cdot)$ that preserves the corresponding property.
We summarise the required assumptions first.

\begin{assumption}\label{ass:stab}
    System~\eqref{eq:sys} is minimal, asymptotically stable, and $\XX:\PP \rightarrow \RR^{n\times n}$ is a solution to~\eqref{eq:stab_LMI}.
\end{assumption}

\begin{assumption}\label{ass:dissipativity}
    System~\eqref{eq:sys} is minimal, $(\QQQ,\SSS,\RRR)$-dissipative with a given $\QQQ : \PP\rightarrow \RR$, $\SSS : \PP\rightarrow \RR$, and $\RRR: \PP\rightarrow \RR$, and $\XX:\PP \rightarrow \RR^{n\times n}$ is a solution to~\eqref{eq:diss_LMI}.
\end{assumption}

The main result of this section for stability- and dissipativity-preserving model reduction is presented next.

\begin{theorem}\label{thm:dissipativity}
    Consider system~\eqref{eq:sys} and matrices $(S,L)$ that satisfy Assumption~\ref{ass:obserAndNoCommonEig}.
    Suppose that $\nu \le n$.
    Consider the reduced-order model~\eqref{eq:prom} with
    \begin{equation}\label{eq:def_G}
        G(p) = \bigl(\Pi(p)^\top \XX(p) \Pi(p)\bigr)^{-1} \Pi^\top(p) \XX(p) B(p),
    \end{equation}
    where the mapping $\Pi(\cdot)$ is the unique solution to~\eqref{eq:SylvesterEq} on $\mathcal{P}$.
    The following statements are true.
    \begin{itemize}
        \item If Assumption~\ref{ass:stab} is satisfied, then the model~\eqref{eq:prom}-\eqref{eq:def_G} is asymptotically stable for all $p\in \PP$.
        \item If Assumption~\ref{ass:dissipativity} is satisfied, then the model~\eqref{eq:prom}-\eqref{eq:def_G} is $(\QQQ,\SSS,\RRR)$-dissipative for all $p\in \PP$ with the same $(\QQQ,\SSS,\RRR)$ as given in Assumption~\ref{ass:dissipativity}.
    \end{itemize}
    In addition, for any $p\in\PP$ such that the condition $\sigma(S)\cap\sigma(S-G(p)L)=\emptyset$ is satisfied, model~\eqref{eq:prom}-\eqref{eq:def_G} also achieves moment matching at $(S,L)$.
\end{theorem}
\begin{proof}
    We only present the proof for dissipativity and moment matching. The proof for asymptotic stability follows a similar line of reasoning and is omitted here for brevity.

    The matrix inequalities~\eqref{eq:diss_LMI} written for the reduced-order model~\eqref{eq:prom} read as
    \begin{multline}\label{eq:proof:LMI_ROM}
        \begin{bmatrix}
            (S-G(p)L)^\top \tilde \XX(p) + \tilde \XX(p) (S-G(p)L) & \star  \\
            G^\top(p) \tilde \XX(p)  & 0 \\
        \end{bmatrix} \\
        - \begin{bmatrix}
            \Pi^\top(p) C^\top(p) \QQQ(p) C(p) \Pi(p) & \star\\ \SSS(p) C(p) \Pi(p) & \RRR(p)
        \end{bmatrix}            
            \preceq 0, \\
        \tilde \XX(p) = \tilde \XX^\top(p) \succ 0, \, \forall p\in \PP,
    \end{multline}
    for some matrix-valued function $\tilde \XX : \PP \rightarrow \RR^{\nu\times\nu}$.
    We show that the selection
    \begin{equation}\label{eq:proof:tildeX}
        \tilde \XX(p) = \Pi^\top(p) \XX(p) \Pi(p)
    \end{equation}
    results in the satisfaction of the matrix inequality~\eqref{eq:proof:LMI_ROM} with the mapping $G(\cdot)$ as in~\eqref{eq:def_G}.

    First, note that $\tilde \XX(p)$ is a positive-definite matrix for any $p\in\PP$ since $\Pi(p)$ is a full column-rank solution of the Sylvester equation~\eqref{eq:SylvesterEq}, \ie, rank$(\Pi(p))=\nu$ for all $p\in\PP$.
    The latter is true by the observability of $(S,L)$, controllability of $(A(\cdot),B(\cdot))$ on $\PP$, and the fact that $\nu \le n$, see~\cite{de1981controllability} for details.

    Using $\tilde \XX(\cdot)$ as in~\eqref{eq:proof:tildeX}, the matrix-valued function $G(\cdot)$ in~\eqref{eq:def_G} can be written as
    \begin{equation}\label{eq:def_G_alternative}
        G(p) = \tilde \XX^{-1}(p) \Pi^\top(p) \XX(p) B(p)
    \end{equation} 
    for all $p\in\PP$.
    By direct computation, we find the following identities for all $p\in\PP$:
    \begin{subequations}\label{eq:proof:identity_F}
        \begin{align}
            \tilde \XX(p) (S-&G(p)L) \\
            &= \tilde \XX(p) (S-\tilde \XX^{-1}(p) \Pi^\top(p) \XX(p) B(p) L) \\
            &= \Pi^\top(p) \XX(p) \Pi(p) S - \Pi^\top(p) \XX(p) B(p)L \\
            &= \Pi^\top(p) \XX(p) A(p) \Pi,\label{eq:proof:use_Sylvester}
        \end{align}
    \end{subequations}
    where, to arrive at~\eqref{eq:proof:use_Sylvester}, the Sylvester equation~\eqref{eq:SylvesterEq} has been used.
    Furthermore, the identity
    \begin{equation}\label{eq:proof:identity_GH}
            \tilde \XX(p) G(p) = \Pi^\top(p) \XX(p) B(p),
    \end{equation}
    can be found directly from~\eqref{eq:def_G_alternative} for all $p\in\PP$.
    Using the identities~\eqref{eq:proof:identity_F} and~\eqref{eq:proof:identity_GH} in~\eqref{eq:proof:LMI_ROM}, the inequality~\eqref{eq:proof:LMI_ROM} can be written as
    \begin{multline}\label{eq:proof:LMI_red}
        \begin{bmatrix}
                \Pi^\top(p) & 0\\
                0 & 1
            \end{bmatrix} 
            \mc K(p) 
            \underbrace{\begin{bmatrix}
                \Pi(p) & 0 \\
                0 & 1
            \end{bmatrix}}_{\eqqcolon \mc T(p)}
            \preceq 0, \\
            \tilde \XX(p) =\tilde \XX^\top(p) \succ 0, \quad \forall p\in\PP,
    \end{multline}
    where $\mc K(\cdot)$ is defined as
    \begin{multline*}
        \mc K(p) \coloneqq \begin{bmatrix}
            A^\top(p) \XX(p) + \XX(p) A(p) & \star  \\
            B^\top(p) \XX(p) & 0 \\
        \end{bmatrix}\\
        - \begin{bmatrix}
            C^\top(p) \QQQ(p) C(p) & \star \\ \SSS(p)C(p) & \RRR(p)
        \end{bmatrix}.
    \end{multline*}
    Now, notice that by Assumption~\ref{ass:dissipativity}, the matrix $\mc K(p)$ satisfies $\mc K(p)\preceq 0$ for all $p\in\PP$.
    Hence, $\mc T^\top(p)\mc K(p)\mc T(p)\preceq 0$ is satisfied for all $p\in\PP$.
    This, together with the fact that $\tilde \XX(\cdot) \succ0$ on $\PP$ implies that~\eqref{eq:proof:LMI_ROM} is satisfied.
    In conclusion, $(\QQQ,\SSS,\RRR)$-dissipativity is preserved for the same $(\QQQ,\SSS,\RRR)$.
    
    Finally, for a fixed $p\in\PP$, it follows from~\cite{astolfi2010model} that moment matching is achieved if in addition $\sigma(S)\cap\sigma(S-G(p)L)=\emptyset$ is satisfied, which completes this proof.
\end{proof}

The result of Theorem~\ref{thm:dissipativity} can be used to construct a stability- or dissipativity-preserving reduced-order model.
First, a matrix-valued function $\XX(\cdot)$ has to be found that satisfies the linear matrix inequalities in Lemma~\ref{lem:stab_diss}, either for stability or dissipativity, depending on the targeted property.
Then, the reduced-order model~\eqref{eq:prom}-\eqref{eq:def_G} that explicitly depends on this $\XX(\cdot)$ preserves the targeted property.

\begin{remark}
By the observability of the pair $(S,L)$ (Assumption~\ref{ass:obserAndNoCommonEig}), a trivial choice of $G(\cdot)$ that results in stability preservation and moment matching is to take $G(p)=\bar G$ with $\bar G\in\mathbb{R}^\nu$ such that $\sigma(S - \bar GL) \cap \sigma(S) = \emptyset$ and $\sigma(S-\bar GL) \subset \mathbb{C}_{<0}$.
Although this trivially enforces asymptotic stability across $\mathcal{P}$, the parametric reduced-order model loses structure in the sense that, \eg, $A(p)$ and $B(p)$ can change with the parameter $p$, while their reduced-order counterparts $S-\bar GL$ and $\bar G$ are then both independent of $p$.
Consequently, this may lead to not capturing the qualitative behaviour of the original system when $p$ varies. 
In contrast, the proposed mapping $G(\cdot)$ in~\eqref{eq:def_G} results in a reduced-order model~\eqref{eq:prom} with a structure that is consistent with the original system~\eqref{eq:sys} in the sense that all its matrices may vary with $p$.    
\end{remark}


\section{Methods to Achieve Parametric Moment Matching for Linear Parametric Systems}\label{sec:MMforLinearParametricSys}
We note from \eqref{eq:prom} that the parametric moment ${\overline{C\Pi}}(\cdot)$ is an instrumental part to be determined to obtain a parametric reduced-order model since the matrix $S$ is prescribed and the mapping $G(\cdot)$ can be selected \textit{a posteriori} to enforce additional properties. In this section, we propose two approaches to approximate the parametric moment ${\overline{C\Pi}}(\cdot)$.
First, in Section~\ref{sec:MatrixBased}, we propose a series expansion to approximate $\overline{C\Pi}(\cdot)$.
Subsequently, in Section~\ref{sec:dataDriven}  we present a second approach to estimate $\overline{C\Pi}(\cdot)$ using basis functions.
In Section~\ref{sec:linearDataDriven}, we provide a data-driven enhancement of the basis function method.  Finally, all the results of this section are illustrated in Section~\ref{sec:linearExample}.

\subsection{Parametric Moment via Series Expansion} \label{sec:MatrixBased}
Consider system \eqref{eq:sys} and suppose (as in, \eg, \cite{benner2015survey}) that $A(\cdot)$ depends on $p$ with a known affine form, namely
\begin{equation} 
    A(p) = A_0 + \sum_{i = 1}^{N_a} f_i^a(p) A_i,\label{eq:Ap}
\end{equation}
where $f_i^a : \mathcal{P} \to \mathbb{R}$ are possibly nonlinear functions that indicate the way $p$ interferes with system \eqref{eq:sys} and $A_i \in \mathbb{R}^{n \times n}$, for $i = 1, \dots , N_a$. Similarly, consider that $B(\cdot)$ and $C(\cdot)$ are described as
\begin{equation} 
    B(p) = B_0 + \sum_{i = 1}^{N_b} f_i^b(p) B_i, \quad C(p) = C_0 + \sum_{i = 1}^{N_c} f_i^c(p) C_i,  \label{eq:BpCp}
\end{equation}
where $f_i^b: \mathcal{P} \to \mathbb{R}$, $f_i^c: \mathcal{P} \to \mathbb{R}$ are possibly nonlinear functions, and $B_i \in \mathbb{R}^{n}$ and $C_i \in \mathbb{R}^{1 \times n}$ are constant matrices, for $i = 1, \dots , N_b$ and $i = 1, \dots , N_c$, respectively.

\begin{assumption}\label{ass:analyticFun}
    The functions $f_i^a(\cdot)$, $f_i^b(\cdot)$, and $f_i^c(\cdot)$ are analytic on $\mathcal{P}$.
\end{assumption}

Under Assumption \ref{ass:analyticFun}, $A(\cdot)$ can be described by using a convergent Taylor series of $f_i^a(\cdot)$ expanded at the origin (without loss of generality), which can be approximated by truncating the higher-order terms
\begin{equation}
\begin{aligned}
    A(p) &= A_0 + \sum_{i = 1}^{N_a} \sum_{j = 0}^{\infty} a_j^i p^j A_i = \sum_{j = 0}^{\infty} p^j \widehat{A}_j \\
    &= \sum_{j = 0}^{M_a} p^j \widehat{A}_j + \mathcal{O}(p^{M_a + 1}),
    \label{eq:ApTaylor}
\end{aligned}
\end{equation}
with $M_a$ the highest retained order, $a_j^i$ the $j$-th Taylor series coefficient of $f_i^a(\cdot)$, and $\widehat{A}_j$ the coefficient of the power series of $A(\cdot)$. Similarly, $B(\cdot)$ and $C(\cdot)$ can be described by
\begin{equation}
    B(p) = B_0 + \sum_{i = 1}^{N_b} \sum_{j = 0}^{\infty} b_j^i p^j B_i = \sum_{j = 0}^{M_b} p^j \widehat{B}_j + \mathcal{O}(p^{M_b + 1}),
    \label{eq:BpTaylor}
\end{equation}
\begin{equation}
    C(p) = C_0 + \sum_{i = 1}^{N_c} \sum_{j = 0}^{\infty} c_j^i p^j C_i = \sum_{j = 0}^{M_c} p^j \widehat{C}_j + \mathcal{O}(p^{M_c + 1}),
    \label{eq:CpTaylor}
\end{equation}
where $M_b$ and $M_c$ denote the highest retained orders for $B(\cdot)$ and $C(\cdot)$, $b_j^i$ and $c_j^i$ are the $j$-th Taylor series coefficients of $f_i^b(\cdot)$ and $f_i^c(\cdot)$, and $\widehat{B}_j$ and $\widehat{C}_j$ are the coefficients of the power series of $B(\cdot)$ and $C(\cdot)$, respectively.

Exploiting the assumptions that we have introduced, we provide a method to compute the parametric moment.

\begin{theorem}\label{thm:approximatePi}
    Consider system~\eqref{eq:sys} and the signal generator~\eqref{eq:SignalGenerator}. Suppose Assumptions~\ref{ass:obserAndNoCommonEig} and~\ref{ass:analyticFun} hold. Then the parametric moment of system~\eqref{eq:sys} at $(S,L)$ on $\mathcal{P}$ is
    \begin{equation}
        \overline{C\Pi}(p) = C(p) \Bigl( \sum_{j = 0}^{\infty} p^j {\Pi}_j \Bigr),  \label{eq:PiTaylor}
    \end{equation}
     with ${\Pi}_j$, $j \geq 0$, the solutions of the nested {Sylvester} equations
    \begin{equation}\label{eq:nestedSylvesterEq}
        \begin{aligned}
            {\Pi}_0 S &= \widehat{A}_0 {\Pi}_0 + \widehat{B}_0 L, \\
            {\Pi}_1 S &= \widehat{A}_0 {\Pi}_1 + \widehat{A}_1 {\Pi}_0 + \widehat{B}_1 L, \\
            &\vdots \\
            {\Pi}_j S &= \sum_{k = 0}^j \widehat{A}_k {\Pi}_{j - k} + \widehat{B}_j L, 
        \end{aligned}
    \end{equation}
    for $j \geq 2$ integer.
\end{theorem}
\begin{proof}
    It follows from Assumption \ref{ass:obserAndNoCommonEig} that the unique mapping $\Pi(\cdot)$ is well defined on $\mathcal{P}$ which solves the parametric Sylvester equation
    \begin{equation}\label{eq:ParaSlvesterEq}
        A(p)\Pi(p) + B(p)L = \Pi(p)S
    \end{equation}
    for all $p \in \mathcal{P}$. By applying a property of the Kronecker product\footnote{$\VEC(AXB)=(B^\top \otimes A)\VEC(X)$, with $A$, $X$, and $B$ of compatible dimensions. \label{footnote:kronProperty}}, the solution to \eqref{eq:ParaSlvesterEq} can be obtained as
    \begin{equation}\label{eq:solutionSylvesterEq}
        \VEC(\Pi(p)) = -(I_\nu \otimes A(p) - S^\top \otimes I_n)^{-1} \VEC(B(p)L),
    \end{equation}
    which, together with Assumption \ref{ass:analyticFun}, implies that the mapping $\Pi(\cdot)$ is also analytic on $\mathcal{P}$ as the sums, products, and reciprocals of analytic functions remain analytic \cite{krantz2002primer}. It follows that $\Pi(\cdot)$ can be formally represented by a convergent power series, namely ${\Pi}(p) = \sum_{j = 0}^{\infty}p^j{\Pi}_j$, locally defined on $\mathcal{P}$, solving~\eqref{eq:ParaSlvesterEq}. Since Assumption \ref{ass:analyticFun} holds, we substitute $A(\cdot)$ and $B(\cdot)$ in \eqref{eq:ParaSlvesterEq} with their power series. Then we have
    \begin{multline}\label{eq:appSylvester}
         \Bigl(p^0\widehat{A}_0 + p^1\widehat{A}_1 + \sum_{j = 2}^{\infty} p^j \widehat{A}_j\Bigr)
        \Bigl(p^0{\Pi}_0 + p^1{\Pi}_1 + \sum_{j = 2}^{\infty} p^j {\Pi}_j\Bigr) \\+ \Bigl(p^0\widehat{B}_0 + p^1\widehat{B}_1 + \sum_{j = 2}^{\infty} p^j \widehat{B}_j\Bigr)L \\= \Bigl(p^0{\Pi}_0 + p^1{\Pi}_1 + \sum_{j = 2}^{\infty} p^j {\Pi}_j\Bigr)S,
    \end{multline}
    and we obtain \eqref{eq:nestedSylvesterEq} by matching the terms of the same order in $p$ on both sides of \eqref{eq:appSylvester}. Finally, by Assumption \ref{ass:obserAndNoCommonEig}, which states that $(S, L)$ is observable, and applying \cite[Lemma 3]{astolfi2010model} for each fixed $p$, the obtained $\overline{C\Pi}(p)$ is in one-to-one relation with the moments as defined in Definition \ref{def:moments} for all $p \in \mathcal{P}$.
\end{proof}

\begin{remark}
    Equation \eqref{eq:solutionSylvesterEq} provides an analytic solution to obtain the exact $\Pi(\cdot)$ without the need to solve the nested Sylvester equations \eqref{eq:nestedSylvesterEq}. However, solving \eqref{eq:solutionSylvesterEq} symbolically is often hindered by the so-called ``expression swell'' problem~\cite{ExpressionSwell}, especially for computations involving high-dimensional elements, which may render this procedure computationally intractable. The value of Theorem \ref{thm:approximatePi} is that it provides a way to solve \eqref{eq:ParaSlvesterEq} numerically, obtaining an approximation of the parametric moment $\overline{C\Pi}(\cdot)$ by truncating the series \eqref{eq:PiTaylor}.
\end{remark}

From a practical point of view, we define an approximate version of the parametric moment by using a finite number of terms ${\Pi}_j$, $j = 0, \dots, N \leq \min \{M_a, M_b\}$ in \eqref{eq:PiTaylor}.

\begin{definition}\label{def:appParaMM}
    We call $\widehat{C\Pi}_N(p) \coloneqq C(p)\widehat{\Pi}_N(p)$, with $\widehat{\Pi}_N(p) \coloneqq \sum_{j = 0}^{N-1} p^j {\Pi}_j$ and $\Pi_j$ computed by Theorem \ref{thm:approximatePi}, the \textit{$N$-th approximate parametric moment} (via series expansion) of system \eqref{eq:sys} at $(S,L)$ on $\mathcal{P}$.
\end{definition}

Definition \ref{def:appParaMM} provides an approximation of the parametric moment $\overline{C\Pi}(\cdot)$ within a neighbourhood of $p = 0$. The associated approximation error with respect to $\overline{C\Pi}(\cdot)$ is exactly quantified by the residual $C(p)\sum_{j = N}^\infty p^j \Pi_j$.

\begin{corollary} \label{thm:convergeTaylorSeries}
    Suppose the assumptions in Theorem \ref{thm:approximatePi} hold. Then $\lim_{N \to \infty} \widehat{C\Pi}_N(p) = \overline{C\Pi}(p)$.
\end{corollary}
\begin{proof}
    The proof directly follows from Theorem \ref{thm:approximatePi}.
\end{proof}

With Definition \ref{def:appParaMM}, the family of parametric {reduced-order} models that match the $N$-th approximate parametric moment of system \eqref{eq:sys} at $(S,L)$ {on} $\mathcal{P}$ is given by
\begin{equation}\label{eq:appParaROM}
    \begin{aligned}
        \Dot{\xi}(t,p) &= (S - {G(p)}L)\xi(t,p) + {G(p)}u(t),\\
        \psi(t,p) &= \widehat{C\Pi}_N(p)\xi(t,p),
    \end{aligned}
\end{equation}
with $G(p)$ such that $\sigma(S - {G(p)}L) \cap \sigma(S) = \emptyset$ for all $p \in \mathcal{P}$ and $\nu < n$. Note that by Corollary \ref{thm:convergeTaylorSeries}, this approximation is with an arbitrary level of accuracy.

We can also design $G(\cdot)$ in~\eqref{eq:appParaROM} to preserve stability or dissipativity using the results of Theorem~\ref{thm:dissipativity}.
To this end, in addition to the mapping $\Pi(\cdot)$ we also require the mapping $\XX(\cdot)$. Similarly to what has been done with $\Pi(\cdot)$,
$\XX(\cdot)$ may be replaced by its $N$-th approximation $\widehat \XX_N(p) \coloneqq \sum_{j = 0}^{N-1} p^j {\XX}_j$, with $\XX_j$ defined 
as the solution of a nested series of equations similar to the ones in Theorem~\ref{thm:approximatePi}.
For instance, for stability, the coefficients $\XX_j$ are the solutions of the nested Lyapunov-like equations
    \begin{equation}\label{eq:nestedLyapunovEq}
        \begin{aligned}
            \widehat{A}_0^\top \mathcal{X}_0 + \mathcal{X}_0 \widehat{A}_0 &= -Q, \\
            \widehat{A}_0^\top \mathcal{X}_1 + \mathcal{X}_1 \widehat{A}_0 &= -\left(\widehat{A}_1^\top \mathcal{X}_0 + \mathcal{X}_0 \widehat{A}_1\right), \\
            &\vdots \\
            \widehat{A}_0^\top \mathcal{X}_j + \mathcal{X}_j \widehat{A}_0 &= - \sum_{k = 1}^j \left(\widehat{A}_k^\top \mathcal{X}_{j-k} + \mathcal{X}_{j-k} \widehat{A}_k\right),
        \end{aligned}
    \end{equation}
for $j \geq 2$ integer. 
The matrices $\widehat{A}_j$ are the same as those in~\eqref{eq:ApTaylor}, while $Q$ can be selected as any positive definite matrix. 
In fact, here $Q$ parametrises the set of stable parametric reduced-order models~\eqref{eq:prom} (or~\eqref{eq:appParaROM}).
Under Assumptions~\ref{ass:stab} and~\ref{ass:analyticFun}, analogously to Theorem~\ref{thm:approximatePi} and Corollary~\ref{thm:convergeTaylorSeries}, it can be shown that $\lim_{N \to \infty} \widehat{\XX}_N(p) =  \XX(p)$.
A similar set of nested equations can be formed for the dissipativity-preserving case. 
We omit the presentation of these equations for reasons of space.

\subsection{Parametric Moment via Basis Functions} \label{sec:dataDriven}
The model-based approach presented in Section \ref{sec:MatrixBased} hinges on the series expansion around a point in $\mathcal{P}$. Hence, as will be shown in Section~\ref{sec:linearExample}, the deviation error between the obtained $\widehat{C\Pi}_N(\cdot)$ and the exact parametric moment ${\overline{C\Pi}}(\cdot)$ typically grows when the parameter value of interest is far away from the expansion point. 
Moreover, the previous method hinges on the knowledge of the expansions~\eqref{eq:Ap}-\eqref{eq:BpCp}, which may be unavailable in many cases of practical interest. In this section, we propose an alternative (model-based) approach to approximate ${\overline{C\Pi}}(\cdot)$ that allows for a more uniform deviation error across the entire parameter space $\mathcal{P}$. A data-driven extension that exploits only input-output time-domain samples is provided in a subsequent section.

Let $\overline{C\Pi}^i(\cdot)$ denote the $i$-th column of ${\overline{C\Pi}}(\cdot)$, and assume\footnote{We use this as a standing assumption throughout this section.}
\begin{equation}\label{eq:standingAssCPiBasisFun}
    \overline{C\Pi}^i(p) = \sum_{j = 1}^M \varphi_j(p) \gamma^i_j,
\end{equation}
where $\varphi_j: \mathcal{P} \to \mathbb{R}$, with $j = 1, \dots, M$ ($M$ may be $\infty$), are bounded basis functions, and $\gamma_j^i \in \mathbb{R}$, with $j = 1, \dots, M$ and $i = 1, \dots, \nu$, are associated constant coefficients. Thus, the parametric moment is expressed in terms of a weighted sum of the basis functions as
\begin{equation} \label{eq:weightedSumCPI}
    {\overline{C\Pi}}(p) = \sum_{j=1}^{M} 
    \varphi_j(p)
     \begin{bmatrix}
         \gamma^1_j & \dots & \gamma^\nu_j
     \end{bmatrix},
\end{equation}
or, in matrix form, as
\begin{equation}\label{eq:compactWeightedSumCPI}
    {\overline{C\Pi}}(p) 
     =  \underbrace{\begin{bmatrix}
        \varphi_1(p) & \dots & \varphi_M(p) \\
    \end{bmatrix}}_{\displaystyle \Phi(p)}
    \underbrace{\begin{bmatrix}
        \gamma^1_1 & \dots & \gamma^\nu_1 \\
        \vdots & \ddots & \vdots \\
        \gamma^1_M & \dots & \gamma^\nu_M \\
    \end{bmatrix}}_{\displaystyle \Gamma},
\end{equation}
where {$\Phi: \mathcal{P} \to \mathbb{R}^{1 \times M}$} is a vector-valued mapping that contains a basis function in each entry and $\Gamma \in \mathbb{R}^{M \times \nu}$ is the associated weight matrix.

In reality, due to either the fact that the full information of the set of basis functions that exactly describe $\overline{C\Pi}^i(\cdot)$ may not be available or the fact that $M$ (the number of basis functions in the set) is too large or even infinite, we consider only a subset of the basis functions $\varphi_j(\cdot)$, with $j = 1, \hdots, N$, and $N \leq M$ finite\footnote{Slightly abusing the notation, the symbol $N$, which was used in the previous section to indicate the number of retained terms $\Pi_j$, also denotes the number of basis functions here to reflect the fact that it remains to be the number of terms in the approximate parametric moment.}. Based on these basis functions, \eqref{eq:compactWeightedSumCPI} can be rewritten as
\begin{equation}\label{eq:appCPIBasisFun}
    {\overline{C\Pi}}(p) = 
    \underbrace{\begin{bmatrix}
        \varphi_1(p) & \dots & \varphi_N(p) \\
    \end{bmatrix}}_{\displaystyle \Phi_N(p)}
    \underbrace{\begin{bmatrix}
        \widetilde{\gamma}^1_1 & \dots &  \widetilde{\gamma}^\nu_1 \\
        \vdots & \ddots & \vdots \\
        \widetilde{\gamma}^1_N & \dots &  \widetilde{\gamma}^\nu_N \\
    \end{bmatrix}}_{\displaystyle \widetilde{\Gamma}_N} + e(p),
\end{equation}
with $e(\cdot)$ the least-squares error resulting by projecting~\eqref{eq:compactWeightedSumCPI} onto the function space identified by $N$ basis functions.

From a practical point of view, we define an approximate version of the parametric moment by using a finite number of basis functions $\varphi_j(\cdot)$, with $j = 1, \hdots, N$, while neglecting the least-squares error $e(\cdot)$.

\begin{definition}\label{def:appParaMMBF}
    We call $\widetilde{C\Pi}_N(p) := \Phi_N(p) \widetilde{\Gamma}_N$, with $\Phi_N(p)$ and $\widetilde{\Gamma}_N$ defined in~\eqref{eq:appCPIBasisFun}, the \textit{$N$-th approximate parametric moment} (via basis functions) of system \eqref{eq:sys} at $(S,L)$ on $\mathcal{P}$.
\end{definition}

Definition \ref{def:appParaMMBF} provides an approximation of the parametric moment $\overline{C\Pi}(\cdot)$ using a finite number of basis functions. Although the approximation error $e(\cdot)$ is not exactly quantifiable, $\widetilde{C\Pi}_N(\cdot)$ is the best approximation (in the sense of the next theorem) on the considered function subspace.

In the following theorem, we provide an approach to uniquely determine $\widetilde{\Gamma}_N$. This result is established by assuming that the set of basis functions $\varphi_j(\cdot)$ is not linearly dependent and that the values of $\overline{C\Pi}(p)$ are available for a ``rich'' enough set of values of $p$, formalised by the following assumption.

\begin{assumption} \label{ass:unisolvantCondition}
    The set of basis functions $\{\varphi_1, \varphi_2, \hdots, \varphi_N\}$ are \textit{unisolvent}\footnote{see \cite[Section 2.4]{davis1975interpolation} for the definition.} on $\mathcal{P}$, \ie{}, the only solution to 
    $$
    \sum_{j=1}^N c_j \varphi_j(p_k) = 0, \quad \forall \, k = 1, \cdots, K \ge N,
    $$
    is the trivial solution $c_1 = c_2 = \hdots = c_N = 0$ for any set $\{p_k\}_{k=1}^K \subset \mathcal{P}$.
\end{assumption}

\begin{theorem}\label{thm:paraMMBasisFuns}
Suppose Assumptions \ref{ass:obserAndNoCommonEig} and \ref{ass:unisolvantCondition} hold. Let the set $\{p_k\}_{k=1}^K \subset \mathcal{P}$ be simple, and the matrices $\Upsilon_K \in \mathbb{R}^{K \times N}$ and $R_K \in \mathbb{R}^{K \times \nu}$ {be defined as}
\begin{equation*}
    \Upsilon_K {:=} 
    \begin{bmatrix}
        \varphi_1(p_1) & \dots & \varphi_N(p_1) \\
        \vdots & \ddots & \vdots \\
        \varphi_1(p_K) & \dots & \varphi_N(p_K) \\
    \end{bmatrix} = 
    \begin{bmatrix}
        \Phi_N(p_1) \\
        \vdots \\
        \Phi_N(p_K) \\
    \end{bmatrix},
\end{equation*}
and 
\begin{equation*}
R_K {:=} 
    \begin{bmatrix}
        {\overline{C\Pi}}(p_1) \\
        \vdots \\
        {\overline{C\Pi}}(p_K) 
    \end{bmatrix}. 
\end{equation*}
Then
\begin{equation}\label{eq:appGammaN}
    \widetilde{\Gamma}_N = (\Upsilon_K^\top \Upsilon_K)^{-1} \Upsilon_K^\top R_K
\end{equation}
is the unique least-squares estimate of $\Gamma$.
\end{theorem}
\begin{proof}
    Assumption \ref{ass:obserAndNoCommonEig} ensures that the data in $R_K$ are well defined.
    Let $\widetilde{\Gamma}_N^i$ and $R_K^i$ be the $i$-th column of $\widetilde{\Gamma}_N$ and $i$-th column of $R_K$, respectively. Then, \eqref{eq:appCPIBasisFun} implies that $\widetilde{\Gamma}_N^i$ can be obtained by solving a least-squares problem\footnote{Recall that $\overline{C\Pi}^i(\cdot)$ denotes the $i$-th column of ${\overline{C\Pi}}(\cdot)$.}
    \begin{equation}\label{eq:leastSquare}
        \begin{aligned}
            \widetilde{\Gamma}_N^i &= \operatorname*{arg\,min}_{\Gamma_x} \, \sum_{k = 1}^{K}\Bigl({\overline{C\Pi}^i}(p_k)-\Phi_N(p_k) \Gamma_x\Bigr)^2 \\
            &= \operatorname*{arg\,min}_{\Gamma_x} \, \| R_K^i - \Upsilon_K \Gamma_x \|_2^2,
        \end{aligned}
    \end{equation}
    for $i = 1, \dots, \nu$, which, by the full column rank of $\Upsilon_K$ implied by Assumption \ref{ass:unisolvantCondition} and the set $\{p_k\}_{k=1}^K$ being simple, has the solution (see, \eg, \cite{bishop2006pattern})
    \begin{equation}\label{eq:optmizerLeastSquare}
        \widetilde{\Gamma}_N^i = (\Upsilon_K^{\top} \Upsilon_K)^{-1}\Upsilon_K^\top R_K^i.
    \end{equation}
    We conclude the proof by repeating the procedure for all $i = 1, \dots, \nu$ resulting in $\widetilde{\Gamma}_N$.
\end{proof}

\begin{remark}\label{rmk:fullColumnRankUpsilon}
    In principle, the unisolvency condition in Assumption \ref{ass:unisolvantCondition} is not restrictive. We note that a wide variety of families of basis functions satisfy this condition, including polynomials, radial basis functions (with strictly positive kernels such as Gaussian and multiquadrics \cite{de2018lectures,fasshauer2007meshfree,mongillo2011choosing}), and Fourier basis functions. However, it is possible in practice that the rank of matrix $\Upsilon _k$ is nevertheless numerically deficient. In this case, we suggest that
    instead of formulating \eqref{eq:appCPIBasisFun} as a least-squares problem, one can alternatively reformulate it as a Ridge regression problem by introducing an additional weight penalty term in the cost function, \ie,
    \begin{equation*}\label{eq:ridgeRegression}
        \begin{aligned}
            \widetilde{\Gamma}_N^i &= \operatorname*{arg\,min}_{\Gamma_x} \, \sum_{k = 1}^{K}\Bigl({\overline{C\Pi}^i}(p_k)-\Phi_N(p_k) \Gamma_x\Bigr)^2 + \sum_{j = 1}^N \lambda_j (\gamma_j^i)^2 \\
            &= \operatorname*{arg\,min}_{\Gamma_x} \, \| R_K^i - \Upsilon_K \Gamma_x \|_2^2 + \Lambda \| \Gamma_x \|_2^2,
        \end{aligned}
    \end{equation*}
    where $\Lambda = \diag(\lambda_1, \dots, \lambda_N)$ contains the regularisation parameters $\lambda_j \geq 0$, for $j = 1, \dots, N$. Then, the optimal weight vector $\widetilde{\Gamma}_N^i$ is obtained by \cite{bishop2006pattern}
    \begin{equation}\label{eq:optmizerRidgeRegression}
        \widetilde{\Gamma}_N^i = (\Upsilon_K^{\top} \Upsilon_K + \Lambda)^{-1}\Upsilon_K^\top R_K^i,
    \end{equation}
    and we select $\Lambda \succ 0$ so that $(\Upsilon_K^{\top} \Upsilon_K + \Lambda)$ is guaranteed to be positive definite and thus invertible. The inclusion of $\Lambda$ can also improve numerical conditioning or robustness.
\end{remark}

Under the standing assumption that the parametric moment $\overline{C\Pi}(\cdot)$ is represented by \eqref{eq:standingAssCPiBasisFun}, we can determine an asymptotic property for $\widetilde{C\Pi}_N(\cdot)$ with $\widetilde\Gamma_N$ computed by \eqref{eq:appGammaN}.

\begin{theorem}\label{thm:convergeBasisFuns}
    Suppose Assumptions \ref{ass:obserAndNoCommonEig} and \ref{ass:unisolvantCondition} hold. Then $\lim_{N \to M}\widetilde{\Gamma}_N = \Gamma$ and $\lim_{N \to M}\widetilde{C\Pi}_N(p) = \overline{C\Pi}(p)$.
\end{theorem}

\begin{proof}
    Assumption \ref{ass:obserAndNoCommonEig} guarantees the well-definedness of ${\overline{C\Pi}(\cdot)}$ and Assumption \ref{ass:unisolvantCondition} ensures the well-definedness of $\widetilde{\Gamma}_N$ (and so $\widetilde{C\Pi}_N(\cdot)$). Observe that~\eqref{eq:appGammaN} is the solution to a least-squares approximation and thus represents a projection of $\Gamma$ onto a function space spanned by $N$ basis functions. Hence, by recalling \eqref{eq:compactWeightedSumCPI}, $\widetilde{\Gamma}_N$ tends to $\Gamma$ as the dimension of the function space $N$ tends to $M$. As a direct consequence, $\lim_{N \to M} \widetilde{C\Pi}_N(p) = \lim_{N \to M} \Phi_N(p) \widetilde\Gamma_N = \Phi(p) \Gamma = \overline{C\Pi}(p)$.
\end{proof}

Thus, the approximate parametric moment via basis functions exhibits the property that as the number of basis functions $N$ increases, it approximates the parametric moment with a better (or at least not worse) performance in the sense of the optimisation problem~\eqref{eq:leastSquare}.

\begin{remark}
    The procedure outlined in this section to approximate the parametric moment via basis functions can similarly be applied to approximate the mapping $\XX(\cdot)$ defined in Section~\ref{sec:dissipativity} for stability- and dissipativity-preserving model reduction.
    Note that instead of approximating the mapping $\overline{C\Pi}(\cdot)$, one could estimate the mapping $\Pi(\cdot)$ directly using the method of basis functions.
    Then, for stability- or dissipativity-preserving model reduction, we can use any combination of series expansions or basis functions to approximate the mappings $\Pi(\cdot)$ and~$\XX(\cdot)$.
\end{remark}

\subsection{Data-Driven Approximation of Parametric Moment}\label{sec:linearDataDriven}
From a model-based perspective, given $A(\cdot), B(\cdot)$ and $C(\cdot)$, the data in $R_K$ in Theorem \ref{thm:paraMMBasisFuns} can be obtained directly by solving \eqref{eq:SylvesterEq} for each value of $p_k$. However, in certain scenarios, only the sampling of the system's inputs and outputs, possibly together with a measurement of the target parameter (\eg, temperature), is available, while the high-order parametric model is unknown and prohibitively expensive to identify. Therefore, it is instrumental to have a data-driven (\ie, model-free) approach that yields the parametric moment without the need for such system knowledge.
In Theorem \ref{thm:moment2steadystate}, it has been noted that each $\overline{C\Pi}(p_k)$ can be characterised by the steady-state response, provided it exists, of the output of system \eqref{eq:sys} at $p_k$ when driven by the signal generator \eqref{eq:SignalGenerator}. This steady-state characterisation allows for a reformulation of the result in Theorem \ref{thm:paraMMBasisFuns} based only on time-domain input-output data and the measurement of the parameter $p$. This approach is explored in this section.

Consider the interconnection of system \eqref{eq:sys} with $p = p_k$ and the signal generator \eqref{eq:SignalGenerator}. We denote a set of sample times $T_i^h = \{t_{i-h+1}, \dots, t_{i-1},t_i\}$, where $0 \leq t_0<t_1<\dots<t_{i-h}<\dots<t_i<\dots<t_q$, with $h>0$ and $q\geq h$. The set $T_i^h$ represents a moving window of $h$ sample times $t_k$. Then we evaluate the signal $\omega$ and the output response $y$ of the interconnected system over $T_i^h$ for each element in $\{p_k\}_{k=1}^K$. We introduce a new assumption.

\begin{assumption}\label{ass:fullRankOmega}
    All the eigenvalues of the matrix $S$ have zero real parts and are simple. The pair $(S, \omega(0))$ is excitable. Furthermore, $\sigma(A(p)) \subset \mathbb{C}_{< 0}$ for all $p \in \mathcal{P}$.
\end{assumption}

Note that Assumption \ref{ass:fullRankOmega} implies the second half of Assumption \ref{ass:obserAndNoCommonEig} due to the requirements on the spectra of $S$ and $A(\cdot)$. It is however a stronger assumption, as it ensures the existence of a well-defined steady-state output response, see point \ref{thm:moment2steadystateP2}) in Theorem \ref{thm:moment2steadystate}. We are ready to give the next result.

\begin{theorem}\label{thm:paraMMDD}
Suppose Assumptions \ref{ass:obserAndNoCommonEig}, \ref{ass:unisolvantCondition}, and \ref{ass:fullRankOmega} hold. Let $\Upsilon_K$ be defined as in Theorem \ref{thm:paraMMBasisFuns}, and the time-snapshots $U_{i,K} \in \mathbb{R}^{hK \times \nu N}$ and $O_{i,K} \in \mathbb{R}^{hK}$, with $K \geq N$ and $h \geq \nu$, be defined as
\begin{equation*}
    U_{i,K} := \Upsilon_K \otimes U_i,
    \quad
    O_{i,K} {:=}
    \begin{bmatrix}
        Y_{i}(p_1) \\
        \vdots \\
        Y_{i}(p_K)
    \end{bmatrix},
\end{equation*}
where
\begin{equation*}
        U_i {:=} 
        \begin{bmatrix}
            \omega(t_{i-h+1}) & \dots & \omega(t_{i-1}) & \omega(t_{i})
        \end{bmatrix}^\top
    \end{equation*}
    and
    \begin{equation*}
        Y_{i}(p_k) {:=} 
        \begin{bmatrix}
            y(t_{i-h+1},p_k) & \dots & y(t_{i-1},p_k) & y(t_{i},p_k)
        \end{bmatrix}^\top .
    \end{equation*}
Then
\begin{equation}\label{eq:hatGammaNi}
    \VEC(\widetilde{\Gamma}_{i,N}^\top) = (U_{i,K}^\top U_{i,K})^{-1} U_{i,K}^\top O_{i,K}
\end{equation}
is an approximation of the weight matrix $\widetilde\Gamma_N$, namely there exists a sequence $\{t_i\}$ such that $\lim_{t_i \to \infty} \widetilde\Gamma_{i,N} = \widetilde\Gamma_N$.
\end{theorem}
\begin{proof}
    By Assumption~\ref{ass:obserAndNoCommonEig}, the parametric moment is well-defined. Under Assumption~\ref{ass:fullRankOmega}, it follows from point \ref{thm:moment2steadystateP2}) of Theorem~\ref{thm:moment2steadystate} that~\eqref{eq:manifold} holds for all $p \in \mathcal{P}$, namely,
    \begin{equation}\label{eq:manifold4AllP}
        {x(t,p)} = \Pi(p)\omega(t) + {\epsilon(t,p)},
    \end{equation}
    with $\epsilon(t,p)$ an exponentially decaying signal in time. Multiplying each side of \eqref{eq:manifold4AllP} by $C(p)$ yields
    \begin{equation}\label{eq:steadyStateResponse}
            y(t,p) = {\overline{C\Pi}}(p)\omega(t) + C(p){\epsilon(t,p)}.
    \end{equation}
    Computing \eqref{eq:steadyStateResponse} at all elements of $T_i^h$ and $\{p_k\}_{k=1}^K$ yields
    \begin{equation*}
        O_{i,K} = \underbrace{\begin{bmatrix}
            {\overline{C\Pi}}(p_1)\omega(t_{i-h+1})\\
            \vdots \\
            {\overline{C\Pi}}(p_1)\omega(t_{i})\\
            \vdots \\
            {\overline{C\Pi}}(p_K)\omega(t_{i-h+1})\\
            \vdots \\
            {\overline{C\Pi}}(p_K)\omega(t_{i})\\
        \end{bmatrix}}_{\displaystyle \VEC(U_i R_K^\top)} 
        + \underbrace{\begin{bmatrix}
            C(p_1){\epsilon(t_{i-h+1},p_1)}\\
            \vdots \\
            C(p_1){\epsilon(t_{i},p_1)}\\
            \vdots \\
            C(p_K){\epsilon(t_{i-h+1},p_K)}\\
            \vdots \\
            C(p_K){\epsilon(t_{i},p_K)}
        \end{bmatrix}}_{\displaystyle E_{i,K}}.
    \end{equation*}
    Then, vectorising the above equation and applying the aforementioned property of the Kronecker product yields
    \begin{equation}\label{eq:OiKVECRK}
        O_{i,K} = (I_K\otimes U_i)\VEC(R_K^\top) + E_{i,K}.
    \end{equation}
    It has been shown in \cite[Lemma~3]{padoan2017geometric} that by the excitability of $(S, \omega(0))$ in Assumption \ref{ass:fullRankOmega} it is trivial to select a time sequence such that $\rank(U_i) = \nu$. From the property $\rank(A\otimes B) = \rank{(A)} \rank{(B)}$, it follows that $(I_K\otimes U_i)$ has full column rank. Thus, we obtain from~\eqref{eq:OiKVECRK} that
    \begin{equation*}
        \VEC(R_K^\top) = ((I_K\otimes U_i)^\top(I_K\otimes U_i))^{-1}(I_K\otimes U_i)^\top {\Delta_{i,K}},
    \end{equation*}
    with ${\Delta_{i,K}} := O_{i,K}-E_{i,K}$. Then, by the transposition property\footnote{$(A\otimes B)^\top = A^\top \otimes B^\top$.}, the mixed-product property\footnote{$(A\otimes B)(C\otimes D)=(AC)\otimes(BD)$, if $A$, $B$, $C$, and $D$ are matrices of such sizes that the matrix products $AC$ and $BD$ can be formed.} and the inverse {property}\footnote{$(A\otimes B)^{-1} = A^{-1} \otimes B^{-1}$, with $A$ and $B$ invertible.}, we have
    \begin{align}\label{eq:VECRK}
            \VEC(R_K^\top) &= ((I_K\otimes U_i^\top)(I_K\otimes U_i))^{-1}(I_K\otimes U_i^\top) {\Delta_{i,K}} \notag \\ 
            &= (I_K\otimes U_i^\top U_i)^{-1}(I_K\otimes U_i^\top) {\Delta_{i,K}} \notag \\
            & = (I_K \otimes (U_i^\top U_i)^{-1})(I_K\otimes U_i^\top) {\Delta_{i,K}} \notag \\
            &= I_K\otimes ((U_i^\top U_i)^{-1}U_i^\top) {\Delta_{i,K}}.
    \end{align}
    Now recall that from Theorem \ref{thm:paraMMBasisFuns}, we have
    \begin{equation}\label{eq:appGammaNThm6}
    \widetilde{\Gamma}_N = (\Upsilon_K^\top \Upsilon_K)^{-1} \Upsilon_K^\top R_K.
    \end{equation}
    Taking the transpose of both sides and using the vectorisation operator and the property of the Kronecker product on~\eqref{eq:appGammaNThm6} yields
    \begin{equation*}
        \begin{aligned}
            \VEC(\widetilde\Gamma_N^\top) &= \VEC(R_K^\top \Upsilon_K (\Upsilon_K^\top \Upsilon_K)^{-1}) \\
            &= (((\Upsilon_K^\top \Upsilon_K)^{-1} \Upsilon_K^\top) \otimes I_{\nu}) \VEC(R_K^\top).
        \end{aligned}
    \end{equation*}
    Substituting $\VEC(R_K^\top)$ with its expression \eqref{eq:VECRK} and using the mixed-product property of the Kronecker product yields
    \begin{equation}
        \begin{aligned}
            \VEC(\widetilde\Gamma_N^\top) &= (((\Upsilon_K^\top \Upsilon_K)^{-1} \Upsilon_K^\top) \otimes I_{\nu}) \VEC(R_K^\top) \\
            &= ((\Upsilon_K^\top \Upsilon_K)^{-1} \Upsilon_K^\top) \otimes ((U_i^\top U_i)^{-1}U_i^\top) {\Delta_{i,K}} \\
            &= ((\Upsilon_K^\top \Upsilon_K)^{-1}\otimes (U_i^\top U_i)^{-1})(\Upsilon_K^\top \otimes U_i^\top) {\Delta_{i,K}} \\
            &= ((\Upsilon_K^\top \otimes U_i^\top)(\Upsilon_K\otimes U_i))^{-1}(\Upsilon_K^\top \otimes U_i^\top) {\Delta_{i,K}}.
        \end{aligned}
    \end{equation}
    Recalling that $U_{i,K}$ is defined as $\Upsilon_K \otimes U_i$, we further simplify
    \begin{equation}\label{eq:VECGammaN}
        \begin{aligned}
            \VEC(\widetilde\Gamma_N^\top) 
            &= (U_{i,K}^\top U_{i,K})^{-1} U_{i,K}^\top {\Delta_{i,K}} \\
            &= (U_{i,K}^\top U_{i,K})^{-1} U_{i,K}^\top(O_{i,K}-E_{i,K}) \\
            &= \VEC(\widetilde{\Gamma}_{i,N}^\top)-(U_{i,K}^\top U_{i,K})^{-1} U_{i,K}^\top E_{i,K}. 
        \end{aligned}
    \end{equation}
    Note that by Assumption \ref{ass:fullRankOmega}, $\omega(t) = e^{St} \omega(0)$ is bounded for all $t \in \mathbb{R}$, and additionally, by the boundedness of $\Upsilon_K$ (directly following from the assumption on the boundedness of basis functions $\varphi_j(\cdot)$), all the elements of $U_{i, K}$ are bounded. Furthermore, consider the facts that each entry of $E_{i,K}$ decays to $0$ as $t_i$ tends to infinity, $U_{i,K}$ is element-wise bounded, and $U_{i,K}^\top U_{i,K}$ is invertible (see Remark \ref{re:fullCollumnRank}). Hence, we have
    \begin{multline*}
        \lim_{t_i \to \infty} (\VEC(\widetilde\Gamma_{i,N}^\top) - \VEC(\widetilde\Gamma_N^\top)) \\ = \lim_{t_i \to \infty} (U_{i,K}^\top U_{i,K})^{-1} U_{i,K}^\top E_{i,K} = 0,
    \end{multline*}
    which completes the proof.
\end{proof}

\begin{remark}\label{re:fullCollumnRank}
    The derivation of \eqref{eq:VECRK} and \eqref{eq:VECGammaN} uses the property of the inverse of a Kronecker product, which requires that $U_i^\top U_i$, $\Upsilon_K^\top \Upsilon_K$, and $U_{i,K}^\top U_{i,K}$ be invertible. The matrix $U_i$ has full column rank by Assumption \ref{ass:fullRankOmega} (see \cite[Lemma~3]{padoan2017geometric}) and Assumption \ref{ass:unisolvantCondition} ensures the full column rank of $\Upsilon_K$. Then, by the property of the Kronecker product, we obtain
    \begin{equation*}
        \rank(U_{i,K}) = \rank(\Upsilon_K \otimes U_i) = \rank(\Upsilon_K) \rank(U_i) = \nu N,
    \end{equation*}
    \ie, $U_{i,K}$ has full column rank.
\end{remark}

\begin{remark}
    Theorem \ref{thm:paraMMDD} provides the least-squares approximation of the weight matrix $\widetilde\Gamma_N$ in a data-driven setting. Equation \eqref{eq:hatGammaNi} presents an optimal unbiased estimator (in the sense of maximum likelihood) when the output data $y(t,p)$ are corrupted by zero-mean Gaussian white noise. For the case when both the input data $\omega(t)$ and the output data $y(t,p)$ are measured under Gaussian white noise, \eqref{eq:hatGammaNi} can be modified as the solution of a total least-squares problem, see \cite{markovsky2007overview}, to ensure the optimal estimation of $\widetilde\Gamma_N$ (in the sense of maximum likelihood). Since the least-squares estimator is consistent\footnote{An estimator is consistent if it converges in probability to the true value of the objective as the sample size tends to infinity.}, it is possible to reduce the impact of noise by augmenting the measured data, \ie, increasing $h$. Another common de-noising technique is regularisation \cite[Chapter~6]{boyd2004convex}, or by using a kernel-based moment matching technique~\cite{moreschini2024nonlinear}. The readers are referred to \cite{zhao2024strategies} for additional strategies to alleviate the impact of more general noise in data-driven moment matching methods.
\end{remark}

With the data-driven approximation $\widetilde{\Gamma}_{i, N}$ in Theorem~\ref{thm:paraMMDD}, we define a data-driven version of the approximate parametric moment via basis functions.

\begin{definition}\label{def:EstimatedparaMoment}
    We call $\widetilde{C\Pi}_{i,N}(p) = \Phi_N(p) \widetilde\Gamma_{i,N}$, with $\Phi_N(p)$ defined in~\eqref{eq:appCPIBasisFun} and $\widetilde\Gamma_{i,N}$ computed by~\eqref{eq:hatGammaNi}, the (data-driven) \textit{$N$-th approximate parametric moment} (via basis functions) of system \eqref{eq:sys} at $(S,L)$ on $\mathcal{P}$.
\end{definition}

The data-driven $N$-th approximate parametric moment enjoys an asymptotic property, as shown next.

\begin{theorem}\label{thm:convergeBasisFunsDD}
    Suppose Assumptions \ref{ass:obserAndNoCommonEig}, \ref{ass:unisolvantCondition}, and \ref{ass:fullRankOmega} hold. Then $\lim_{t_i \to \infty} \Bigl({\overline{C\Pi}}(p) - \lim_{N \to M}\widetilde{C\Pi}_{i,N}(p)\Bigr) = 0$.
\end{theorem}
\begin{proof}
    Under Assumptions \ref{ass:obserAndNoCommonEig}, \ref{ass:unisolvantCondition}, and \ref{ass:fullRankOmega}, Theorem \ref{thm:paraMMDD} yields
    \begin{equation*}
        \lim_{t_i \to \infty} \Bigl({\overline{C\Pi}}(p) - \lim_{N \to M}\widetilde{C\Pi}_{i,N}(p)\Bigr) = {\overline{C\Pi}}(p) - \lim_{N \to M}\widetilde{C\Pi}_{N}(p)
    \end{equation*}
    which, under the same assumptions, converges to $0$ by Theorem \ref{thm:convergeBasisFuns}.
\end{proof}

With Definition \ref{def:EstimatedparaMoment}, the family of parametric {reduced-order model}s that match the approximate parametric moment of {system} \eqref{eq:sys} at $(S,L)$ {on} $\mathcal{P}$ is given by
\begin{equation}\label{eq:appParaROMDD}
    \begin{aligned}
        \Dot{\xi}{(t,p)} &= (S - {G(p)}L)\xi{(t,p)} + {G(p)}u(t),\\
        \psi{(t,p)} &= \widetilde{C\Pi}_{i,N}(p)\xi{(t,p)},
    \end{aligned}
\end{equation}
with ${G(p)}$ any matrix such that $\sigma(S - {G(p)}L) \cap \sigma(S) = \emptyset$ and $\nu < n$.

\subsection{Illustration of the Linear Results}\label{sec:linearExample}
    We illustrate the results developed so far using a parametric benchmark model \cite{morwiki_synth_pmodel}. This parametric model is described by the equations
    \begin{equation}
    \begin{aligned}
        \Dot{x}{(t,p)} &= (A_0 + p A_1)x{(t,p)} + Bu(t),\\
        y{(t,p)} &= Cx{(t,p)},\label{eq:syntheticModel}
    \end{aligned}
    \end{equation}
    where $A_0 \in \mathbb{R}^{n \times n}$, $A_1 \in \mathbb{R}^{n \times n}$, $B \in \mathbb{R}^{n}$, and $C \in \mathbb{R}^{1 \times n}$. In this example, we consider the same setting as that in \cite[Chapter~9]{baur2017comparison}, namely,
    \begin{align*}
        A_0 &= 
        \begin{bmatrix}
            A_{0,1} \\
            & \ddots \\
            & & A_{0,k}
        \end{bmatrix}, & A_1 &= 
        \begin{bmatrix}
            A_{1,1}  \\
            & \ddots \\
            & & A_{1,k}
        \end{bmatrix}, \\
        B &=
        \begin{bmatrix}
            B_1 & \dots & B_k
        \end{bmatrix}^\top, & C &= 
        \begin{bmatrix}
            C_1 & \dots & C_k
        \end{bmatrix},
    \end{align*}
    where
    \begin{equation*}
        A_{0,i} = \begin{bmatrix}
            0 & b_i \\
            -b_i & 0
        \end{bmatrix}, A_{1,i} = \begin{bmatrix}
            a_i & 0 \\
            0 & a_i
        \end{bmatrix}, B_i^\top = \begin{bmatrix}
            2 \\ 0
        \end{bmatrix}, C_i^\top = \begin{bmatrix}
            1 \\ 0
        \end{bmatrix},
    \end{equation*}
    for $i = 1, \dots, k$ with $k = 500$ and, hence, $n = 1000$. The coefficients $a_i$ ($b_i$, respectively) are equidistantly spaced in the interval $[-10^3, -10]$ ($[10, 10^3]$, respectively). The parameter space $\mathcal{P}$ is set as $[0.1, 1]$, and the expansion point is selected as $p = 0.55$ which is the middle point of the parameter space.

    To illustrate the results in Section~\ref{sec:MatrixBased}, we select $50$ distinct interpolation points (and their conjugates), sampled equidistantly on a logarithmic scale from $10^0\iota$ to $10^{3.1}\iota$, to reduce system \eqref{eq:syntheticModel} to a model of dimension $\nu = 100$, as done in~\cite[Chapter~9]{baur2017comparison}. The matrix $S$ is constructed in the real Jordan form and $L$ is constructed such that $(S, L)$ is observable. 
    The mapping $G(\cdot)$ is given by~\eqref{eq:def_G} using the approximations $\widehat{\Pi}_N(\cdot)$ and $\widehat{\mathcal{X}}_N(\cdot)$, determined by~\eqref{eq:nestedSylvesterEq} and~\eqref{eq:nestedLyapunovEq}, respectively, with $N = 4$ and $Q = I_n$.
    Although in~\eqref{eq:def_G}, $\Pi^\top(p)\mathcal{X}(p)\Pi(p)$ is theoretically invertible on $\mathcal{P}$, in this example this quantity is ill-conditioned for some $p\in\mathcal{P}$.
    Therefore, we add a regularisation term, replacing $\Pi^\top(p)\mathcal{X}(p)\Pi(p)$ in~\eqref{eq:def_G} with $\Pi^\top(p)\mathcal{X}(p)\Pi(p) +\epsilon I$ where $\epsilon = 10^{-14}$.
    We obtain the $N$-th approximate parametric moment $\widehat{C\Pi}_{N}(\cdot)$ for different values of $N$ by following the computations outlined in Theorem \ref{thm:approximatePi}. 
    The stability-preserving parametric reduced-order model that achieves parametric moment matching at $(S, L)$ on $\mathcal{P}$ is then given by \eqref{eq:appParaROM}.

    To illustrate the results in Section~\ref{sec:linearDataDriven}, and to show the flexibility and effectiveness of the proposed data-driven approach\footnote{Note that since Theorem \ref{thm:paraMMDD} is a data-driven extension of Theorem \ref{thm:paraMMBasisFuns}, this example covers both approaches.
    }, we select only $16$ interpolation points (and their conjugates), sampled equidistantly on a logarithmic scale from $10^0\iota$ to $10^{3.1}\iota$, reducing system~\eqref{eq:syntheticModel} to a model of order $\nu = 32$. The matrix $S$ is constructed using the real Jordan form and $L$ is constructed such that $(S, L)$ is observable.
    To further show the advantages of this approach, we select $G(\cdot) = \bar G$ constant such that the eigenvalues of the reduced-order model~\eqref{eq:appParaROMDD} can be trivially located in the open complex left half-plane by the observability of the pair $(S, L)$. 
    We compute the approximate weight matrix $\widetilde{\Gamma}_{i,N}$ by following Theorem~\ref{thm:paraMMDD}.
    The moving window $T_i^h$ is selected to start from $t_{i-h+1} = 17.38$~s and end with $t_i = 20 $~s, containing $64$ equidistant time instants in between (\ie, $h = 64$). The set $\{p_k\}_{k=1}^K$ comprises points equally sampled on $\mathcal{P} = [0.1, 1]$ with $K = 10$.
    We use polynomial basis functions $\varphi_j(p)=p^{(j-1)}$, for $j = 1, \dots, N$, to obtain the interpolation matrix $\Phi_N(\cdot)$.
    The data-driven $N$-th approximate parametric moment via basis functions $\widetilde{C\Pi}_{i,N}(\cdot)$ is computed by following Definition~\ref{def:EstimatedparaMoment} and then the resulting parametric reduced-order model~\eqref{eq:appParaROMDD} is constructed.

    Fig.~\ref{fig:CPIpDiffNumTerms} shows the relative ${\ell_2}$-norm errors\footnote{This is obtained by solving $\overline{C\Pi}(p^*)$ through \eqref{eq:SylvesterEq} and substituting $p^*$ into $\widehat{C\Pi}_N(p^*)$ or $\widetilde{C\Pi}_{i,N}(p^*)$ for each $p^*$ densely sampled across $\mathcal{P}$.} between $\overline{C\Pi}(\cdot)$ and $\widehat{C\Pi}_N(\cdot)$ (top) and $\overline{C\Pi}(\cdot)$ and $\widetilde{C\Pi}_{i,N}(\cdot)$ (bottom) over $\mathcal{P}$ for different selected values of $N$. For the series expansion method, the approximation is very accurate around the expansion point $p = 0.55$, whilst the deviation increases when $|p-0.55|$ grows. For the basis function method, the approximation is uniformly good (with relative ${\ell}_2$-norm error less than $8 \times 10^{-3}$) across the parameter space. In both graphs, as the number of retained series expansion terms or basis functions increases, the relative ${\ell}_2$-norm error decreases.

    Fig.~\ref{fig:BodeP} depicts the magnitude of the frequency responses of the benchmark model (solid/blue line), of the reduced-order model via series expansion with $\widehat{C\Pi}_4(\cdot)$ (dashed/red line), and of the reduced-order model via basis functions with $\widetilde{C\Pi}_{i,6}(\cdot)$ (dotted/green line) at $p=0.1$ (top left), $p=0.2$ (top right), $p=0.4$ (bottom left), and $p=1.0$ (bottom right). Circles\footnote{We plot $13$ out of $50$ moments here to maintain visual clarity.} and diamonds indicate the moments matched at the desired interpolation frequencies for the two reduced-order models, respectively.  
    We observe that the frequency response of both reduced-order models is almost indistinguishable from that of the original model, 
    which indicates that the behaviour of system~\eqref{eq:syntheticModel} is well approximated by the resulting reduced-order models. 

    To compare the two preceding reduced-order models, which were obtained in different settings, Fig.~\ref{fig:SE&BF} shows the relative $\mathcal{H}_2$-norm error (as defined in \cite[Chapter~9]{baur2017comparison}) between the original benchmark model and the resulting reduced-order models for $\widehat{C\Pi}_4(\cdot)$ (dashed/red line) and $\widetilde{C\Pi}_{i, 6}(\cdot)$ (dotted/green line), respectively. 
    The figure shows that the reduced-order model with $\widetilde{C\Pi}_{i,6}(\cdot)$ approximates the original system with relatively small errors across the entire parameter space, which is remarkable given its smaller order, \ie, order $32$ versus order $100$ of the first reduced-order model characterised by $\widehat{C\Pi}_4(\cdot)$. However, this first model provides an improved approximation for values of $p$ that are close to the expansion point $p = 0.55$ (approximately in the range $[0.35, 0.75]$).

    \begin{figure}[!t]
    \centerline{\includegraphics[width=\columnwidth]{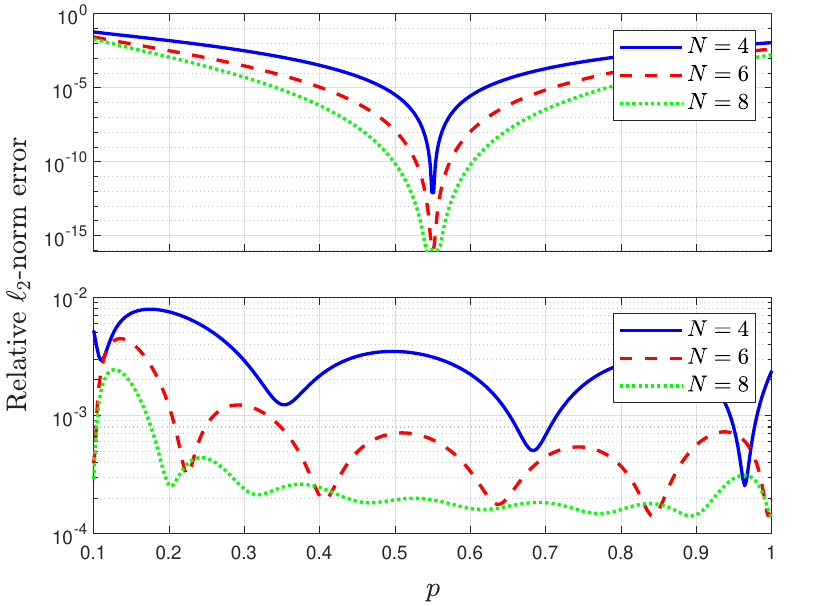}}
    \caption{The relative ${\ell}_2$-norm error, versus the parametric moment ${\overline{C\Pi}}(\cdot)$, of the $N$-th approximate parametric moment $\widehat{C\Pi}_N(\cdot)$ (top) and $\widetilde{C\Pi}_{i, N}(\cdot)$ (bottom), for different values of $N$.}
    \label{fig:CPIpDiffNumTerms}
    \end{figure}

    \begin{figure}[!t]
    \centerline{\includegraphics[width=\columnwidth]{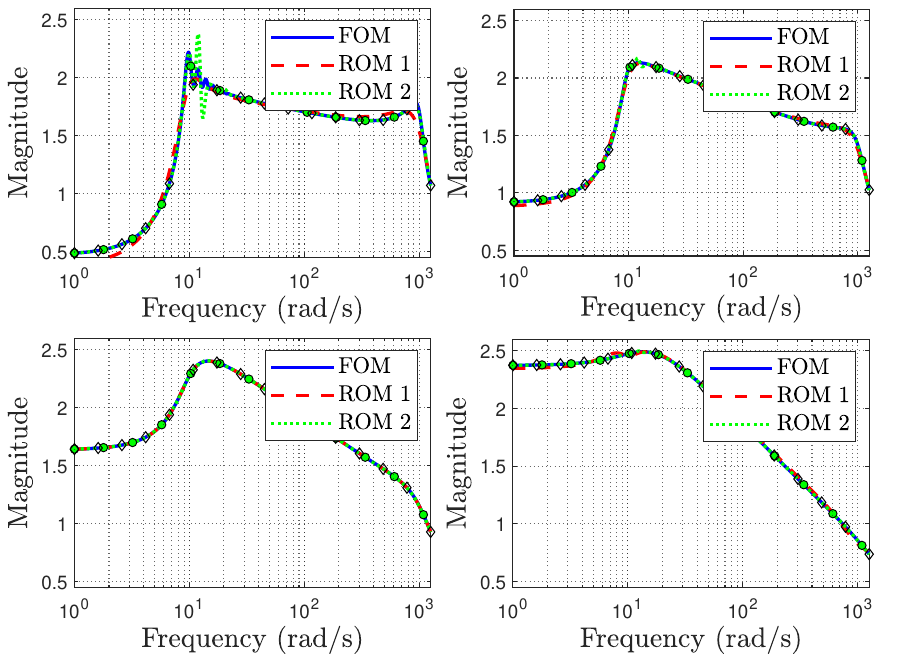}}
    \caption{Frequency responses (magnitude) of the benchmark model (solid/blue line), of the reduced-order model by series expansion with $\widehat{C\Pi}_4(\cdot)$ (dashed/red line), and of the reduced-order model via basis functions with $\widetilde{C\Pi}_{i,6}(\cdot)$ (dotted/green line) at $p=0.1$ (top left), $p=0.2$ (top right), $p=0.4$ (bottom left), and $p=1.0$ (bottom right). Circles and diamonds indicate the moments matched at the desired interpolation frequencies for the two reduced-order models, respectively.}
    \label{fig:BodeP}
    \end{figure}

    \begin{figure}[!t]
    \centerline{\includegraphics[width=\columnwidth]{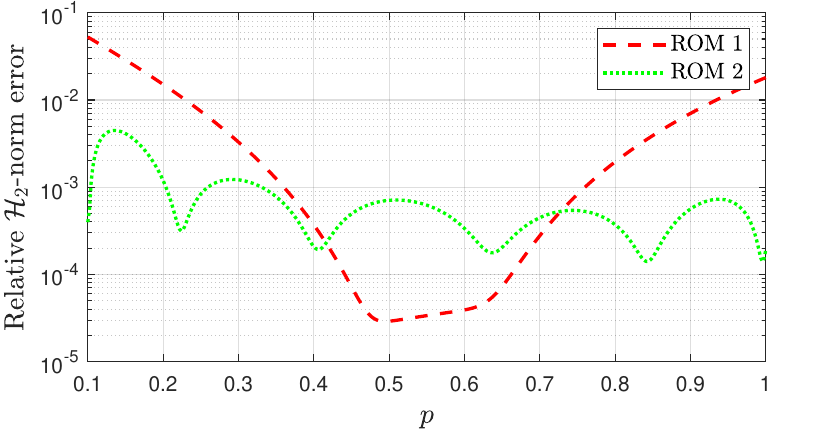}}
    \caption{The relative $\mathcal{H}_2$-norm error between the original benchmark model and the resulting reduced-order models for $\widehat{C\Pi}_4(\cdot)$ (dashed/red line) and $\widetilde{C\Pi}_{i, 6}(\cdot)$ (dotted/green line), respectively.}
    \label{fig:SE&BF}
    \end{figure}

In summary, we have proposed two different approaches, namely series expansion and basis functions, for approximating the parametric moment. On the one hand, the approach via series expansion typically provides an accurate approximation of the parametric moment in a neighbourhood of the expansion point and thus better suits a parameter space $\mathcal{P}$ with a small range; on the other hand, the approach via basis functions typically provides an accurate approximation of the parametric moment across the entire parameter space $\mathcal{P}$, with larger errors around the nominal value, but smaller errors at the edges of $\mathcal{P}$. The approach via basis functions has then been extended to a data-driven setting, which allows using the method even when knowledge of the high-order parametric system is limited.

\section{Moment Matching for Nonlinear Parametric Systems}\label{sec:MMforNonlinearParametricSys}

In this section, we first recall the notion of moment and moment matching for nonlinear non-parametric systems. Analogously to the linear case, we extend these notions to nonlinear parametric systems in Section~\ref{sec:notionNonlinearMoment} by introducing the notion of nonlinear parametric moment. Then, we propose a data-driven approach to approximate the nonlinear parametric moment in Section~\ref{sec:nonDataDriven}. Finally, in Section~\ref{sec:nonlinearExample}, we illustrate the results of this section using a wind farm model.

\subsection{Notion of Moment and Moment Matching} \label{sec:notionNonlinearMoment}
Consider now a nonlinear, SISO, continuous-time parametric system described by
\begin{equation} 
    \begin{aligned}
        \Dot{x}(t,p) &= f(x(t, p), u(t),p),\\
        y(t,p) &= h(x(t, p),p),\label{eq:nonlinearsys}
    \end{aligned}
\end{equation}
with parameter $p \in \mathcal{P} \subseteq \mathbb{R}$, state $x(t,p) \in \mathbb{R}^{n}$, input $u(t) \in \mathbb{R}$, output $y(t,p) \in \mathbb{R}$, $f: \mathbb{R}^{n} \times \mathbb{R} \times \mathcal{P}  \to \mathbb{R}^{n}$ and $h: \mathbb{R}^{n} \times \mathcal{{P}} \to \mathbb{R}$ smooth mappings parametrised by $p$ with $f(0, 0, p) = 0$, $h(0, p) = 0$ for all $p \in \mathcal{P}$, and consider a signal generator
\begin{equation} \label{eq:nonlinearSignalGenerator}
    \dot{\omega}(t) = s(\omega(t)), \quad u(t) = l(\omega(t)), 
\end{equation}
where $\omega(t) \in \mathbb{R}^\nu$, $s: \mathbb{R}^{\nu} \to \mathbb{R}^{\nu}$ and $l: \mathbb{R}^{\nu} \to \mathbb{R}$ are smooth mappings with $s(0) = 0$ and $l(0) = 0$. The interconnection of the signal generator \eqref{eq:nonlinearSignalGenerator} and the nonlinear parametric system~\eqref{eq:nonlinearsys} is described by
\begin{equation}
    \begin{aligned}
        \dot{\omega}(t) &= s(\omega(t)), \\
        \Dot{x}(t,p) &= f(x(t, p), l(\omega(t)),p),\\
        y(t,p) &= h(x(t, p),p).\label{eq:nonlinearInterconnectedSys}
    \end{aligned}
\end{equation}
We define the notion of parametric moment for system \eqref{eq:nonlinearsys} by modifying the notion of moment for nonlinear systems proposed in \cite{scarciotti2017nonlinear}, which we recall in the following for one fixed parameter $p=p^*$ after introducing essential assumptions.

\begin{assumption}\label{ass:nonlinearPiExist}
    There is a unique mapping $\pi({\omega}, p^*)$, locally defined in a neighbourhood $\mathcal{N}$ of ${\omega} = 0$, with $\pi(0, p^*) = 0$, which solves the partial differential equation (PDE)
    \begin{equation}\label{eq:nonlinearPDE}
        \frac{\partial \pi({\omega},p^*)}{\partial {\omega}}s({\omega}) = f(\pi({\omega}, p^*), l({\omega}), p^*). 
    \end{equation}
\end{assumption}

\begin{assumption}\label{ass:nonlinearSGObservable}
    The signal generator \eqref{eq:nonlinearSignalGenerator} is observable\footnote{See \cite{scarciotti2017data} for the definition of observable systems.}.
\end{assumption}

The combination of Assumptions \ref{ass:nonlinearPiExist} and \ref{ass:nonlinearSGObservable} is a nonlinear counterpart of Assumption \ref{ass:obserAndNoCommonEig}.

\begin{definition}[see \cite{scarciotti2017nonlinear}]
    Consider system \eqref{eq:nonlinearsys} and the signal generator \eqref{eq:nonlinearSignalGenerator}. Suppose Assumptions \ref{ass:nonlinearPiExist} and \ref{ass:nonlinearSGObservable} hold. The mapping $h \circ \pi$ is the (nonlinear) moment of system \eqref{eq:nonlinearsys} at $(s, l)$.
\end{definition}

Analogously to the linear case, we define the parametric moment for system \eqref{eq:nonlinearsys} by considering the dependence on the parameter $p$.

\begin{definition}\label{def:nonlinearParaMoment}
    We call the mapping $\kappa(\omega, p) := h \circ \pi(\omega, p)$ the \textit{(nonlinear) parametric moment} of system~\eqref{eq:nonlinearsys} at $(s, l)$ on $\mathcal{P}$, where $\pi: \mathcal{N} \times \mathcal{P} \to \mathbb{R}^n$ is a unique mapping solving \eqref{eq:nonlinearPDE} for all $p \in \mathcal{P}$.
\end{definition}

Note that if Assumptions \ref{ass:nonlinearPiExist} and \ref{ass:nonlinearSGObservable} hold for all $p \in \mathcal{P}$, then $h \circ \pi$ exists and is unique for each $p \in \mathcal{P}$. Consequently, under the same assumptions, the nonlinear version of the parametric moment is well defined.

With Definition \ref{def:nonlinearParaMoment}, a family of parametric reduced-order models that match the parametric moment of system \eqref{eq:nonlinearsys} at $(s,l)$ on $\mathcal{P}$ are given by \cite{scarciotti2017nonlinear}
\begin{equation}\label{eq:nonlinearROM}
    \begin{aligned}
        \dot{\xi}(t, p) &= s(\xi) - \delta(\xi, p)l(\xi) + \delta(\xi, p) u(t), \\
        \psi(t,p) &= \kappa(\xi, p),
    \end{aligned}
\end{equation}
where $\delta$ is any mapping such that the equation 
\begin{multline*}
    \frac{\partial q(\omega,p)}{\partial \omega} s(\omega)=s(q(\omega,p)) - \delta(q(\omega,p), p) l(q(\omega,p)) \\+ \delta(q(\omega,p), p) l(\omega)
\end{multline*}
has the unique solution $q(\omega,p)=\omega$ for all $p \in \mathcal{P}$.

\subsection{Parametric Moment via Basis Functions} \label{sec:nonDataDriven}
Similar to the linear case, determining the parametric moment of system \eqref{eq:nonlinearsys} requires solving equation \eqref{eq:nonlinearPDE} for all $p \in \mathcal{P}$, which is challenging. A possible solution consists in extending the ideas of Section~\ref{sec:MatrixBased} by expanding $f$, $\pi$, $s$, $l$, and so the PDE~\eqref{eq:nonlinearPDE}, with respect not only to $p$, but also to $x$ and $\omega$. Then one can solve the resulting infinitely many generalised Sylvester equations in both $p$ and $\omega$. However, this method would be computationally demanding, in addition to being accurate only around the expansion point $(p^*, \omega^*)$. In this section, we instead develop a nonlinear enhancement of the results in Section \ref{sec:linearDataDriven}. As the resulting method is data-driven, it also has the advantage of being usable when the mappings $f$, $g$, and $h$ are unknown. To this end, we begin with an assumption.

\begin{assumption} \label{ass:nonlinearBasisSpace}
    The mapping $\kappa$ is in the function space identified by the set of continuous bounded basis functions $\eta_j: \mathcal{N} \times \mathcal{P} \to \mathbb{R}$, with $j = 1, \dots, Z$ ($Z$ may be $\infty$), \ie{}, there exist constant coefficients $\theta_j \in \mathbb{R}$, with $j = 1, \dots, Z$, such that $\kappa(\omega,p) = \sum_{j = 1}^Z \theta_j \eta_j({\omega},p)$ for any pair $({\omega},p)$.
\end{assumption}

Let
\begin{equation}\label{eq:nonAppMMMatrixForm}
    \begin{aligned}
        \widetilde{\Theta}_N & :=
        \begin{bmatrix}
            \widetilde\theta_1 & \dots & \widetilde\theta_{N}
        \end{bmatrix}^\top,\\
        H_{N}({\omega},p) & :=
        \begin{bmatrix}
            \eta_1({\omega},p) & \dots & \eta_{N}({\omega},p)
        \end{bmatrix},
    \end{aligned}
\end{equation}
with $N \leq Z$. By Assumption \ref{ass:nonlinearBasisSpace}, the parametric moment of system \eqref{eq:nonlinearsys} can be described by a weighted sum of basis functions
\begin{equation}\label{eq:nonParaMomentLeastSquare}
    \begin{aligned}
        \kappa(\omega,p) &= \sum_{j = 1}^{N} \widetilde\theta_j \eta_j({\omega},p) + \widetilde e({\omega},p)\\
        &= H_{N}({\omega},p) \widetilde{\Theta}_{N} + \widetilde e({\omega},p),
    \end{aligned}
\end{equation}
where $\widetilde e(\omega, p)$ is the least-squares error resulting by projecting $\kappa$ into a function space identified by $N < Z$ basis functions. Analogously to the linear case, we define an approximation of the nonlinear parametric moment by neglecting the least-squares error $\widetilde e(\omega, p)$.

\begin{definition}\label{def:nonAppParaMMBF}
    We call $\widetilde\kappa_{N}(\omega,p) := H_{N}({\omega},p)\widetilde \Theta_{N}$, with $\widetilde \Theta_{N}$ and $H_{N}({\omega},p)$ defined in~\eqref{eq:nonAppMMMatrixForm}, the \textit{$N$-th approximate (nonlinear) parametric moment} (via basis functions) of system~\eqref{eq:nonlinearsys} at $(s,l)$ on $\mathcal{P}$.
\end{definition}

Before proposing a data-driven method to estimate the weight vector $\widetilde \Theta_{N}$, an additional assumption (which is the nonlinear version of Assumption \ref{ass:fullRankOmega}) is introduced for system \eqref{eq:nonlinearsys} and the signal generator \eqref{eq:nonlinearSignalGenerator} to connect the notion of moment with the steady-state output response of the interconnected system \eqref{eq:nonlinearInterconnectedSys}.

\begin{assumption}\label{ass:nonStableInterconnection}
    System \eqref{eq:nonlinearsys} is minimal and locally exponentially stable at the zero equilibrium point for any possible value of $p \in \mathcal{P}$. The signal generator \eqref{eq:nonlinearSignalGenerator} is neutrally stable\footnote{See \cite[Chapter~8]{isidori1995nonlinear} for the definition of neutrally stable systems.} and satisfies the excitation rank condition\footnote{See \cite{padoan2017geometric} for the definition of the excitation rank condition.} at $\omega(0)$. In addition, the initial condition $\omega(0)$ of the signal generator is almost periodic\footnote{See \cite{padoan2017geometric} for the definition of an almost periodic point.} and all the solutions of the signal generator are analytic.
\end{assumption}

\begin{lemma}\label{lm:nonMoment2SteadyState}
    Consider system~\eqref{eq:nonlinearsys} and the signal generator~\eqref{eq:nonlinearSignalGenerator}. Suppose Assumptions \ref{ass:nonlinearSGObservable} and \ref{ass:nonStableInterconnection} hold. Then Assumption \ref{ass:nonlinearPiExist} holds for all $p \in \mathcal{P}$. In addition, the steady-state output response of the interconnected system \eqref{eq:nonlinearInterconnectedSys} is $y_{ss}(t, p) = h(\pi({\omega}, p),p)$ for any $x(0)$ and $\omega(0)$ sufficiently small.
\end{lemma}
\begin{proof}
    This is a straightforward consequence of~\cite[Lemma 2.1]{scarciotti2017nonlinear} applied to each $p \in \mathcal{P}$.
\end{proof}

Lemma \ref{lm:nonMoment2SteadyState} shows that the parametric moment of system~\eqref{eq:nonlinearsys} can be characterised by the steady-state output response of the interconnected system \eqref{eq:nonlinearInterconnectedSys}. This observation opens the possibility of approximating the parametric moment from samples of $\omega$ and $y$ collected at a set of distinct values of $p$. To ensure the uniqueness of the estimated weight vector, the following nonlinear counterpart of Assumption \ref{ass:unisolvantCondition} is introduced.
\begin{assumption} \label{ass:unisolvantConditionNonlinear}
    The set of basis functions $\{\eta_1, \eta_2, \hdots, \eta_N\}$ are \textit{unisolvent}, \ie{}, the only solution to 
    $$
    \sum_{j=1}^N c_j \eta_j(\omega(t_i), p_k) = 0
    $$
    for all $k = 1, \hdots, K$ and $i = 1, \hdots, h$, is the trivial solution $c_1 = c_2 = \hdots = c_N = 0$ for any $\{\omega(t_i), p_k\}_{i=1, k=1}^{h, K} \subset \mathcal{N} \times \mathcal{P}$.
\end{assumption}

We are now ready to provide a data-driven method that can be used to approximate the nonlinear parametric moment.

\begin{theorem}\label{thm:nonlinearDDWeightVector}
    Suppose Assumptions \ref{ass:nonlinearSGObservable}, \ref{ass:nonlinearBasisSpace}, \ref{ass:nonStableInterconnection}, and \ref{ass:unisolvantConditionNonlinear} hold. Let $T_i^h$ be the time window at which the signal ${\omega}$ and the output response $y(t,p)$ of the interconnected system \eqref{eq:nonlinearInterconnectedSys} are evaluated on the simple set $\{p_k\}_{k=1}^K$. Define the time-snapshots $R_{i,K} \in \mathbb{R}^{hK\times N}$ and $M_{i,K} \in \mathbb{R}^{hK}$, with $h \geq \nu$ and $hK \geq N$, as
    \begin{equation*}
        R_{i,K} \coloneqq
        \begin{bmatrix}
            R_i(p_1)\\
            \vdots \\
           R_i(p_K)
        \end{bmatrix}, \quad
        M_{i,K} \coloneqq
        \begin{bmatrix}
            M_i(p_1)\\
            \vdots \\
            M_i(p_K)
        \end{bmatrix},
    \end{equation*}
    with
    \begin{equation*}
        R_i(p_k) \coloneqq 
        \begin{bmatrix}
            H_{N}(\omega(t_{i-h+1}),p_k)^\top & \dots & H_{N}(\omega(t_{i}),p_k)^\top
        \end{bmatrix}^\top
    \end{equation*}
    and 
    \begin{equation*}
        M_i(p_k) \coloneqq 
        \begin{bmatrix}
            y(t_{i-h+1}, p_k) & \dots & y(t_{i}, p_k)
        \end{bmatrix}^\top.
    \end{equation*}
    If $R_{i,K}$ has full column rank, then
    \begin{equation}\label{eq:nonApproximateTL}
        \VEC(\widetilde \Theta_{i,N}) = (R_{i,K}^\top R_{i,K})^{-1}R_{i,K}^\top M_{i,K}
    \end{equation}
    is an approximation of the weight vector $\widetilde \Theta_{N}$, namely there exists a sequence $\{t_i\}$ such that $\lim_{t_i \to \infty}\widetilde \Theta_{i,N} = \widetilde \Theta_{N}$.
\end{theorem}
\begin{proof}
    Under Assumptions \ref{ass:nonlinearSGObservable} and \ref{ass:nonStableInterconnection}, Lemma \ref{lm:nonMoment2SteadyState} indicates that the output response of system \eqref{eq:nonlinearInterconnectedSys} can be written as
    \begin{equation}\label{eq:nonSSOutputResponse}
        y(t,p) = y_{ss}(t,p) + \widetilde\epsilon(t,p) = \kappa(\omega,p) + \widetilde\epsilon(t,p),
    \end{equation}
   where $\widetilde\epsilon(t,p)$ is an exponentially decaying signal in time. Substituting~\eqref{eq:nonParaMomentLeastSquare} into~\eqref{eq:nonSSOutputResponse} yields
    \begin{equation}\label{eq:nonSSOutputResponseDivided}
        y(t,p) = H_N({\omega},p)\widetilde \Theta_{N} + \widetilde e({\omega},p) + \widetilde\epsilon(t,p).
    \end{equation}
    Computing~\eqref{eq:nonSSOutputResponseDivided} at all elements of $T_i^h$ and $\{p_j\}_{j=1}^K$ yields
    \begin{equation}
        M_{i,K} = R_{i,K}\widetilde \Theta_{N} + \underbrace{\begin{bmatrix}
            \widetilde e(\omega(t_{i-h+1}),p_1)\\
            \vdots \\
            \widetilde e(\omega(t_{i}),p_1)\\
            \vdots \\
            \widetilde e(\omega(t_{i-h+1}),p_K)\\
            \vdots \\
            \widetilde e(\omega(t_{i}),p_K)\\
        \end{bmatrix}}_{\displaystyle D_{i,K}} 
        + \underbrace{\begin{bmatrix}
            \widetilde\epsilon(t_{i-h+1},p_1)\\
            \vdots \\
            \widetilde\epsilon(t_{i},p_1)\\
            \vdots \\
            \widetilde\epsilon(t_{i-h+1},p_K)\\
            \vdots \\
            \widetilde\epsilon(t_{i},p_K)
        \end{bmatrix}}_{\displaystyle \widetilde E_{i,K}}.
    \end{equation}
    As $R_{i,K}$ has full column rank by assumption, the weight vector $\widetilde \Theta_{N}$ is the solution of the least-squares problem
    \begin{equation}\label{eq:nonLeastSquare}
        \widetilde \Theta_{N} = \operatorname*{arg\,min}_{\Theta_x} \, (M_{i,K} - \widetilde E_{i,K} - R_{i,K} \Theta_x),
    \end{equation}
    with $D_{i,k}$ the truncated least-squares error. It follows from the full column rank of $R_{i,K}$ (implied by Assumption~\ref{ass:unisolvantConditionNonlinear} and the set $\{p_k\}_{k=1}^{K}$ being simple) that  
    $
        \VEC(\widetilde \Theta_{N}) = (R_{i,K}^\top R_{i,K})^{-1}R_{i,K}^\top (M_{i,K} - \widetilde E_{i,K}). 
    $
    Note that the error vector $\widetilde E_{i,K}$ decays to $0$ when $t_i$ goes to infinity, $R_{i,K}$ is bounded element-wise (by the boundedness of the basis functions in Assumption \ref{ass:nonlinearBasisSpace}) and $R_{i,K}^\top R_{i,K}$ is invertible (by the full column rank of $R_{i,K}$), hence we have
    \begin{equation*}
        \lim_{t_i \to \infty}(\widetilde \Theta_{i,N} - \widetilde \Theta_{N}) = \lim_{t_i \to \infty}(R_{i,K}^\top R_{i,K})^{-1}R_{i,K}^\top \widetilde E_{i,K} = 0,
    \end{equation*}
    which completes the proof.
\end{proof}

We are now ready to introduce the definition of the data-driven approximate nonlinear parametric moment of system~\eqref{eq:nonlinearsys}.
\begin{definition}\label{def:nonEstimatedparaMoment}
    We call $\widetilde\kappa_{i,N}(\omega,p) = H_{N}({\omega},p)\widetilde \Theta_{i,N}$, with $H_{N}({\omega},p)$ defined in~\eqref{eq:nonAppMMMatrixForm} and $\widetilde \Theta_{i,N}$ computed by~\eqref{eq:nonApproximateTL}, the (data-driven) \textit{$N$-th approximate (nonlinear) parametric moment} (via basis functions) of system~\eqref{eq:nonlinearsys} at $(s,l)$ on $\mathcal{P}$.
\end{definition}

The data-driven $N$-th approximate nonlinear parametric moment enjoys an asymptotic property, as shown below.

\begin{theorem}
    Suppose Assumptions \ref{ass:nonlinearSGObservable}, \ref{ass:nonlinearBasisSpace}, \ref{ass:nonStableInterconnection}, and \ref{ass:unisolvantConditionNonlinear} hold. Then $\lim_{t_i \to \infty}\bigl(\kappa(\omega,p) - \lim_{N \to Z}\widetilde\kappa_{i,N}(\omega,p)\bigr)=0$.
\end{theorem}
\begin{proof}
    If Assumptions \ref{ass:nonlinearSGObservable} and \ref{ass:nonStableInterconnection} hold, then the parametric moment $\kappa$ of system \eqref{eq:nonlinearsys} is well-defined on $\mathcal{P}$. In addition if Assumptions~\ref{ass:nonlinearBasisSpace} and~\ref{ass:unisolvantConditionNonlinear} hold, by Theorem~\ref{thm:nonlinearDDWeightVector},
    \begin{equation*}
        \lim_{t_i \to \infty}\bigl(\kappa(\omega,p) - \lim_{N \to Z}\widetilde\kappa_{i,N}(\omega,p)\bigr) = \kappa(\omega,p) - \lim_{N \to Z}\widetilde\kappa_{N}(\omega,p),
    \end{equation*}
    and by Assumption \ref{ass:nonlinearBasisSpace}, $\lim_{N \to Z}\widetilde\kappa_{N}(\omega,p) {=} \kappa(\omega,p)$.
\end{proof}

With Definition \ref{def:nonEstimatedparaMoment}, a family of parametric reduced-order models that match the $N$-th approximate nonlinear parametric moment of system \eqref{eq:nonlinearsys} at $(s,l)$ on $\mathcal{P}$ is given by
\begin{equation}\label{eq:appNonlinearROM}
    \begin{aligned}
        \dot{\xi}(t, p) &= s(\xi) - \delta(\xi, p)l(\xi) + \delta(\xi, p) u(t), \\
        \psi(t,p) &= \widetilde\kappa_{i,N}(\xi, p),
    \end{aligned}
\end{equation}
where $\delta$ is any mapping\footnote{A mapping with such a property is, for instance, any $\delta$ that renders the origin of \eqref{eq:appNonlinearROM} asymptotically stable in the first approximation.} such that the equation 
\begin{multline*}
    \frac{\partial q(\omega,p)}{\partial \omega} s(\omega)=s(q(\omega,p)) - \delta(q(\omega,p), p) l(q(\omega,p)) \\+ \delta(q(\omega,p), p) l(\omega)
\end{multline*}
has the unique solution $q(\omega,p)=\omega$ for all $p \in \mathcal{P}$.

\subsection{Illustration of the Nonlinear Results}\label{sec:nonlinearExample}
    We consider a $20\, \times\, 1$ MW wind farm on a~$66$~kV distribution feeder as a practical testing environment. The basic layout of the wind farm used in this example is shown in Fig. \ref{fig:layout_windFarm}. The triangles represent Type $4$ wind turbine generators (WTGs), and the wind farm of interest comprises~$2$ strings of WTGs ($10$ WTGs on each string). These WTGs are connected to the distribution feeder through a wind farm transformer. The feeder side of the wind farm transformer is denoted as the point of common coupling (PCC), assuming that some other loads or generators may be connected to the point in real-world scenarios. The dynamical model of each WTG, including the impedance on the generator side, is of dimension $17$, and the readers are referred to \cite{gong2024model,zhang2025a} for details.

    \begin{figure}[!t]
    \centerline{\includegraphics[width=\columnwidth]{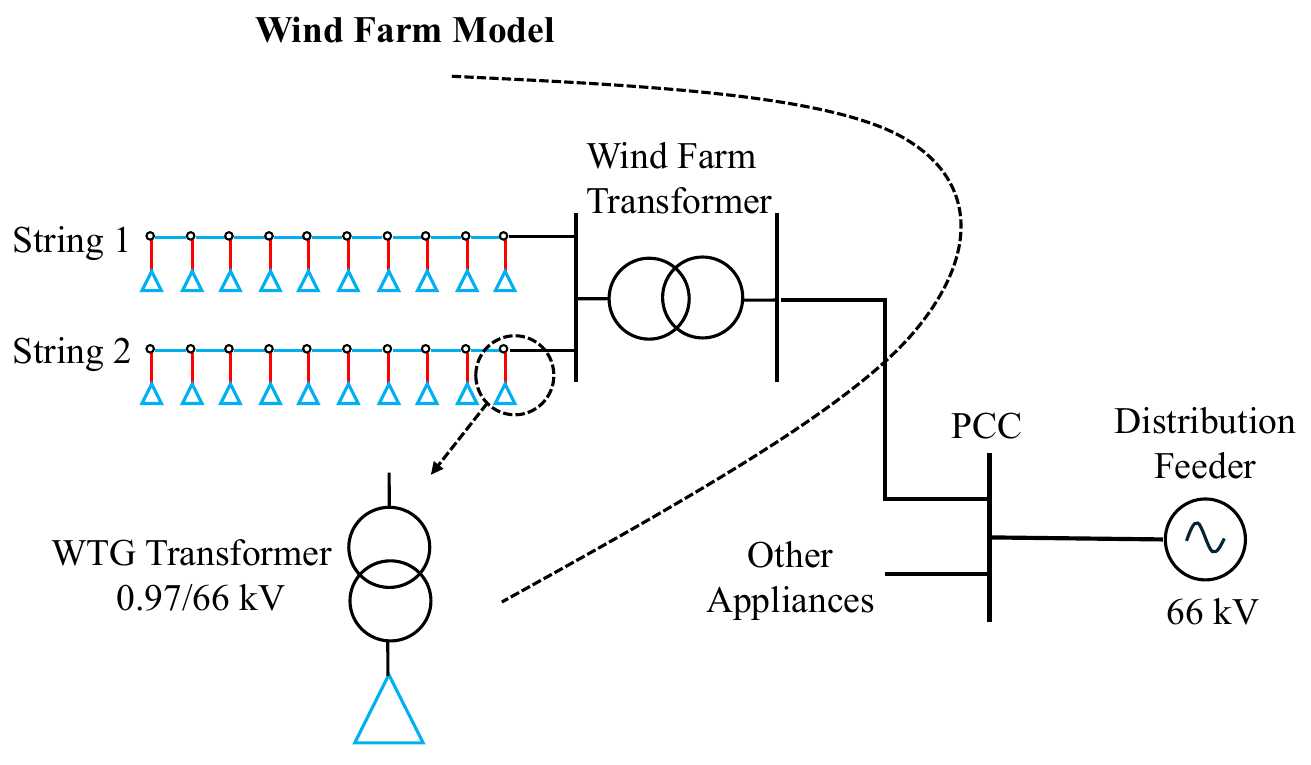}}
    \caption{The basic layout of a wind farm connected to a distribution feeder~\cite{zhang2025a}.}
    \label{fig:layout_windFarm}
    \end{figure}

    The entire wind farm is modelled using Matlab/Simulink resulting in a nonlinear system of order $344$. The input and output of the system are selected as the root mean square (rms) value of the PCC voltage and the active power of the wind farm, respectively. The system is considered parametrised by the wind speed which varies from $6$ to $10$ m/s.
    
    In theory, the input signal is typically assumed to be constant at $1$ pu. However, in reality, the loads and generators connected to PCC will cause voltage quality issues, \eg, slow voltage variations and flicker \cite{tande2002applying}. These issues can result in the input signal oscillating around the nominal value at multiple frequencies. Therefore, in this example, the input signal is modelled as the sum of the constant nominal value of $1$ pu and a series of sinusoidal perturbations\footnote{Such modelling procedure is reasonable because the PCC voltage can be measured and then the dominant frequencies of the voltage signal can be obtained by the Fourier transform.}. Specifically, three sinusoids of periods $24$ h ($f_1 = 1.157 \times 10^{-5}$ Hz), $40$~min ($f_2 = 0.025$ Hz) and $7$~min ($f_3 = 0.143$ Hz), and of amplitudes within $0.1$ pu are used to model slow voltage variations\footnote{The values are chosen according to \cite{tande2002applying}, namely the slow voltage variations are allowed to be within $\pm 10 \%$ of the nominal value.}; and two sinusoids of frequencies $f_4 = 0.11$ Hz and $f_5 = 0.2$ Hz, and of amplitudes within $0.05$ pu are used to model flickers\footnote{The values are chosen according to \cite{tande2002applying} and \cite{dejamkhooy2014modeling} in which it is shown that flickers are generally within $\pm 5 \%$ of the nominal value and of frequencies from $0.1$ to $35$ Hz.}.
    
    To ensure that the parametric reduced-order model can capture the dynamical behaviour of the wind farm in this scenario of interest, we consider the following signal generator
    \begin{equation}\label{eq:nonSignalGeneratorWindFarm}
        \dot \omega(t) = S \omega(t), \quad u(t) = L \omega(t),
    \end{equation}
    where $S$ is constructed in real Jordan form such that
    $\sigma(S) = \{0, \pm 2\pi f_i\}_{i=1}^{5}$, 
    and $L \in \mathbb{R}^{1 \times 11}$ is selected such that $(S, L)$ is observable. The signal $\omega$ and the output $y$ are evaluated over $T_i^h$, where $t_{i-h+1} = 148.67$ s and $t_i = 200$ s, with $h = 440$. The number of wind speed samples is selected as $K = 17$, and these samples are equally spaced on $\mathcal{P} = [6, 10]$. With these settings, we use Theorem \ref{thm:nonlinearDDWeightVector} to compute the parametric moment by choosing $N = 40$ 
    radial basis functions (RBFs) with Gaussian kernel $\eta_j(\omega, p)=\exp{(-\|[\omega\quad p]^\top - c_j\|_2^2/(2\tilde{\sigma}_j^2))}$, all with the same $\tilde{\sigma}_j = 1$ and with centres $c_j$ randomly sampled from the space $[2\omega_{\text{min}}-\omega_{\text{max}}, 2\omega_{\text{max}}-\omega_{\text{min}}] \times [2p_{\text{min}}-p_{\text{max}}, 2p_{\text{max}}-p_{\text{min}}]$, where $\omega_{\text{min}}$ and $\omega_{\text{max}}$ ($p_{\text{min}}$ and $p_{\text{max}}$) represent the minimum and maximum values of the $\omega$ samples (wind speed samples), respectively. The obtained parametric moment is utilised to construct the parametric reduced-order model (of order $\nu = 11$) in the form of \eqref{eq:appNonlinearROM} with $\delta$ a constant vector whose entries are all equal to $220$. Such selection of $\delta$ guarantees the stability of the reduced-order model.

    We simulate the output responses of the full-order model (FOM), \ie, the wind farm, and of the obtained reduced-order model (ROM) in the scenario where voltage quality issues occur at $t = 80$ s, when the systems are in their steady states with nominal input $1$ pu. Fig. \ref{fig:output_windFarm} depicts the output responses of the FOM (solid/blue line) and the ROM (dashed/red line) for various values of the wind speed $p$. It is emphasised that the time required to simulate the FOM at a single parameter value is around $5$ hours while the time to simulate the ROM is only $1$ second, which elucidates the need of model reduction for such large-scale systems. Fig.~\ref{fig:error_windFarm} shows the errors between the output response of the FOM and that of the ROM, \ie, $y(t, p) - \psi(t, p)$, for the corresponding wind speed values. It is noted that the ROM captures well the steady-state behaviour of the FOM with steady-state errors within $-0.1$ to $0.1$ pu, not only for the sampled wind speed $p = 7.5$, but also for the wind speeds $p = 6.4$ and $p = 9.3$ that \textit{do not} belong to the sampled set.

    \begin{figure}[!t]
    \centerline{\includegraphics[width=\columnwidth]{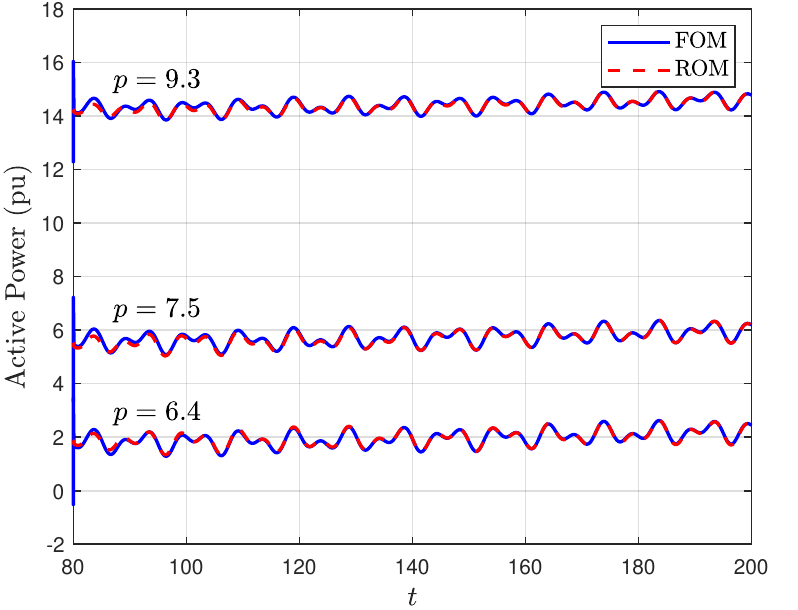}}
    \caption{Time histories of the output response of the full-order model (FOM) (solid/blue line) and of the output response of the obtained parametric reduced-order model (ROM) (dashed/red line) for various wind speed values.}
    \label{fig:output_windFarm}
    \end{figure}

    \begin{figure}[!t]
        \centerline{\includegraphics[width=\columnwidth]{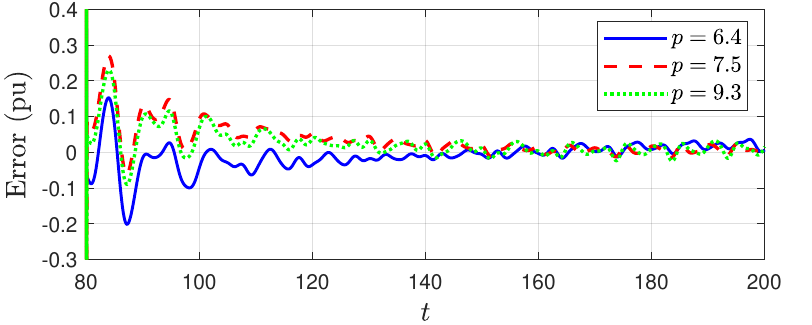}}
        \caption{Errors between the output response of the full-order model (FOM) and the output response of the obtained parametric reduced-order model (ROM) for three different wind speed values.}
        \label{fig:error_windFarm}
    \end{figure}

\section{Conclusion}\label{sec:conclusion}
We have presented a series of methods to address the problem of model reduction by moment matching for linear and nonlinear parametric systems. In doing so, the notion of moment has been extended to both linear and nonlinear parametric cases. 
Furthermore, we have provided a method for preserving the asymptotic stability and dissipativity properties of the original system.
For linear systems, approaches that approximate the parametric moment via series expansion and basis functions have been proposed. 
In addition, we have provided a data-driven enhancement of the approach via basis functions. 
For nonlinear systems, we have proposed a data-driven approach via basis functions to approximate the nonlinear parametric moment. Exploiting these approximations, families of parametric reduced-order models have been constructed. 
Finally, the use of the proposed approaches has been illustrated through simulations.


\section*{References}
\bibliographystyle{IEEEtran}
\bibliography{ref}

\begin{IEEEbiography}[{\includegraphics[width=1in,height=1.25in,clip,keepaspectratio]{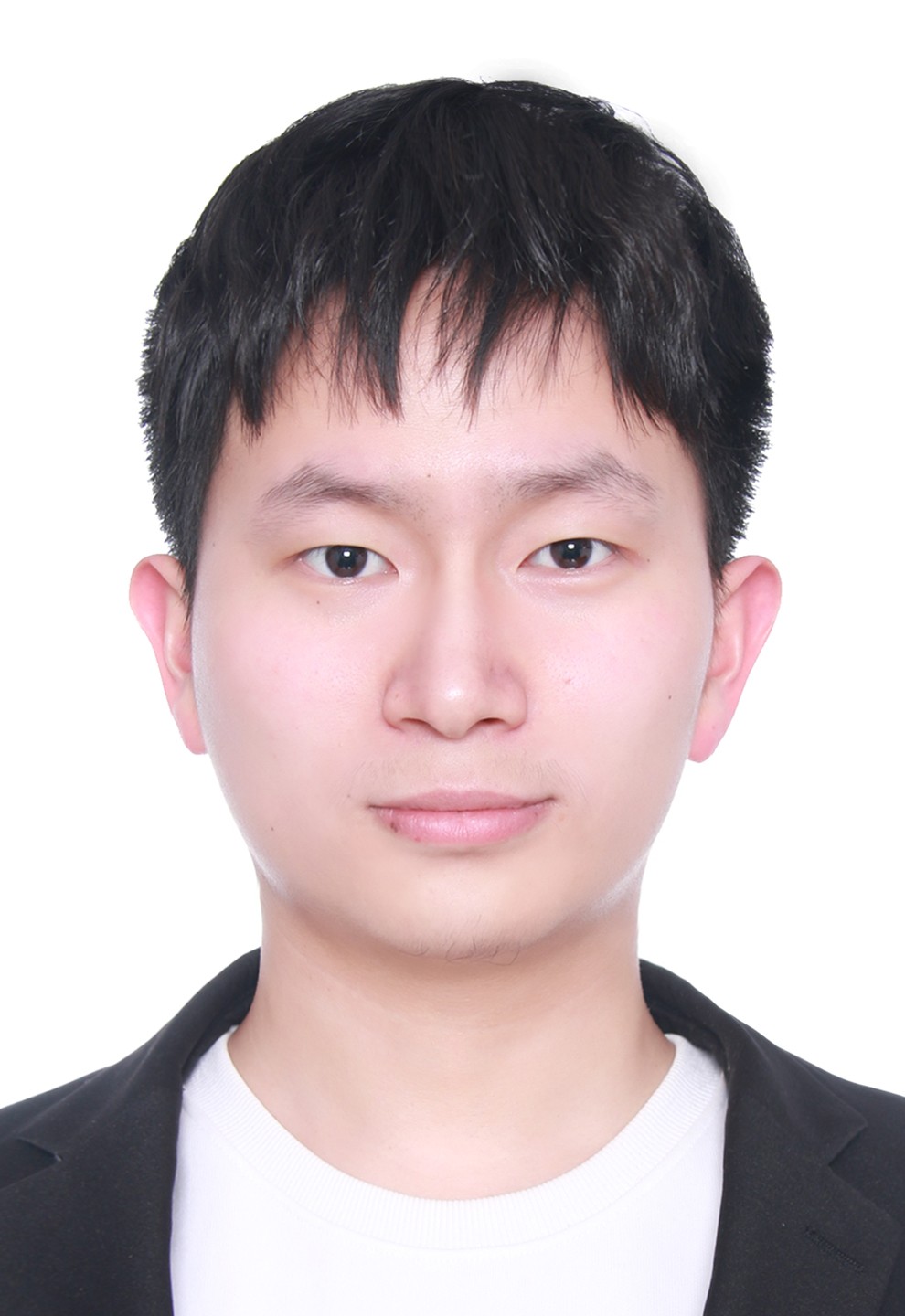}}]{Hanqing Zhang} (Student Member, IEEE) was born in Jiangsu, China, in 1999. He received his B.Eng. degree in Detection, Guidance, and Control from Harbin Engineering University, China, in 2021, and M.Sc. degree in Control and Optimisation from Imperial College London, UK, in 2022. He is currently pursuing his Ph.D. degree in the MAC-X Lab, Control and Power Group, at Imperial College London. His current research interests are focused on model reduction and power systems.
\end{IEEEbiography}

\begin{IEEEbiography}[{\includegraphics[width=1in,height=1.25in,clip,keepaspectratio]{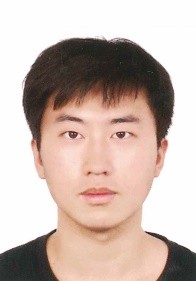}}]{Junyu Mao} (Student Member, IEEE) is a Ph.D. candidate at the MAC-X Lab, Imperial College London. He received his B.Eng. in Electrical and Electronic Engineering from the University of Liverpool in 2018, and two M.Sc. degrees: one in Control Systems from Imperial College London (2019), and another in Data Science and Machine Learning from University College London (2020). In 2021, he joined the Control and Power Group at Imperial College London to pursue his Ph.D. degree. His research focuses on the theoretical foundations of reduced-order modelling and data-driven control for large-scale dynamical systems.
\end{IEEEbiography}

\begin{IEEEbiography}[{\includegraphics[width=1in,height=1.25in,clip,keepaspectratio]{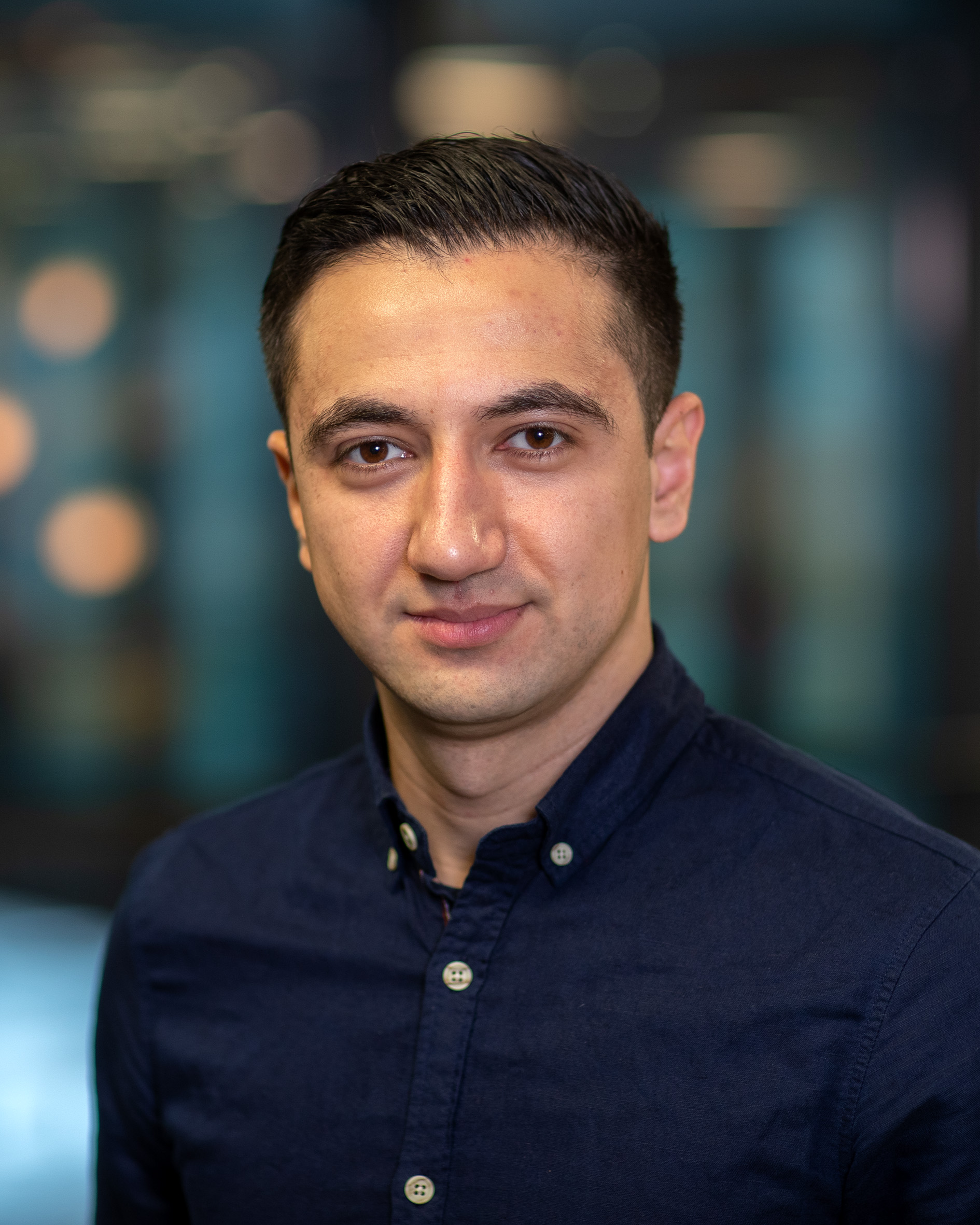}}]{Fahim Shakib} (Member, IEEE) is a Research Associate in the Control and Power Group at Imperial College London. He received both his M.Sc. (cum laude) and Ph.D. (cum laude) in Mechanical Engineering from Eindhoven University of Technology, with his doctoral work focused on data-driven modelling and complexity reduction for nonlinear systems with stability guarantees. His research interests include system identification, model reduction, and control of nonlinear systems.
\end{IEEEbiography}

\begin{IEEEbiography}[{\includegraphics[width=1in,height=1.25in,clip,keepaspectratio]{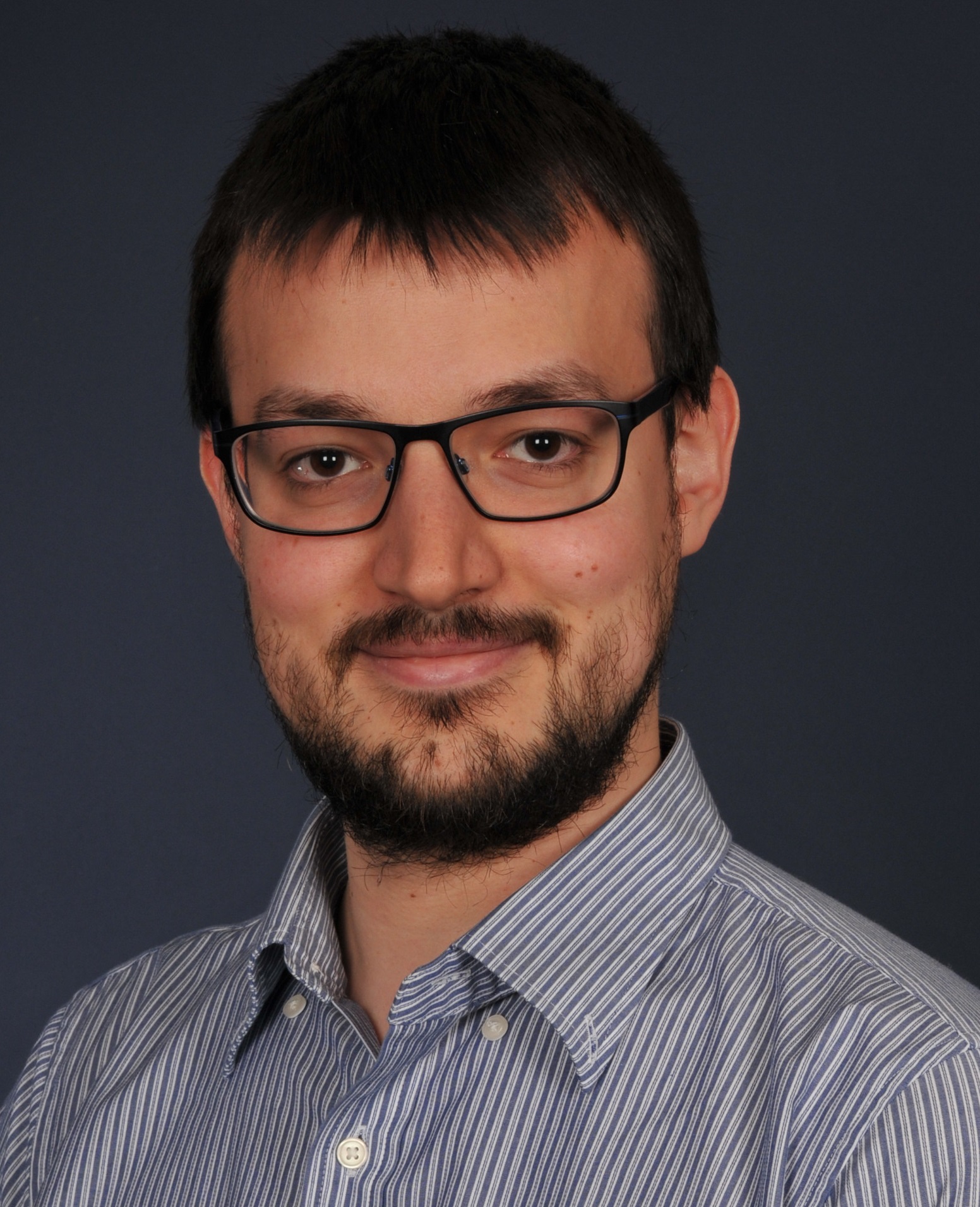}}]{Giordano Scarciotti} (Senior Member, IEEE) received his B.Sc. and M.Sc. degrees in Automation Engineering from the University of Rome “Tor Vergata”, Italy, in 2010 and 2012, respectively, and his Ph.D in Control Engineering and M.Sc. in Applied Mathematics from Imperial College London, UK, in 2016 and 2020, respectively.
He is currently an Associated Professor at Imperial. He was a visiting scholar at New York University in 2015 and at University of California Santa Barbara in 2016 and a Visiting Fellow of Shanghai University in 2021-2022. He is the recipient of the IET Control~\& Automation PhD Award (2016), the Eryl Cadwaladr Davies Prize (2017), an ItalyMadeMe award (2017), and the IEEE Transactions on Control Systems Technology Outstanding Paper Award (2023). He is a member of the EUCA CEB, and of the IFAC and IEEE CSS TCs on Nonlinear Control Systems. He is Associate Editor of Automatica. He was the NOC Chair for the ECC 2022 and of the 7th IFAC Conference on Analysis and Control of Nonlinear Dynamics and Chaos 2024, and the Invited Session Chair and Editor for the IFAC Symposium on Nonlinear Control Systems 2022 and 2025, respectively.
\end{IEEEbiography}

\end{document}